\documentclass[american,aps,pra,reprint,nofootinbib]{revtex4-2}
\usepackage[T1]{fontenc}
\usepackage{color}
\usepackage[utf8]{inputenc}
\usepackage{amsmath}
\usepackage{xcolor}
\usepackage{ragged2e}
\usepackage[singlelinecheck=false]{caption}
\usepackage{subcaption}
\usepackage{bbold}
\usepackage{graphicx}
\usepackage{bbm}
\usepackage{babel}
\usepackage{amsmath}
\usepackage{amssymb}
\usepackage{amsthm}
\usepackage{physics}
\usepackage{braket}
\usepackage{amssymb}
\usepackage{appendix}
\usepackage{dcolumn}
\usepackage{times,txfonts}
\usepackage{float}

\theoremstyle{plain}
\newtheorem{definition}{Definition}[section]
\theoremstyle{plain}
\newtheorem{thm}{Theorem}
\newtheorem{low}{Lower bound}
\usepackage[unicode=true,pdfusetitle,bookmarks=true,bookmarksnumbered=false,bookmarksopen=false,
breaklinks=false,pdfborder={0 0 0},pdfborderstyle={},backref=false,colorlinks=true]{hyperref}
\hypersetup{allcolors=magenta}

\DeclareCaptionTextFormat{hyphen}{\justifying\hyphenpenalty=50\exhyphenpenalty=50 #1}
\usepackage[singlelinecheck=false]{caption}
\begin{document}
\title{ENTCALC: Toolkit for calculating geometric entanglement in multipartite quantum systems}

\author{Piotr Masajada}
\author{Aby Philip}
\author{Alexander Streltsov}
\email{streltsov.physics@gmail.com}
\affiliation{Institute of Fundamental Technological Research, Polish Academy of Sciences, Pawi\'{n}skiego 5B, 02-106 Warsaw, Poland}

\begin{abstract}
    We present \textbf{entcalc}, a Python and MATLAB package for estimating the geometric entanglement of multipartite quantum states. The package operates as follows: given a multipartite quantum state as input, it outputs an estimate of its geometric entanglement. For pure states, it computes the geometric entanglement together with an estimation error. For mixed states, it provides both lower and upper bounds on the geometric entanglement, thereby identifying an interval in which the true value lies. We provide several methods to compute the lower bound, enabling users to balance accuracy against computational cost.
We apply entcalc to several representative examples, including for 
$3\otimes3$ PPT entangled states, mixtures of GHZ and W states, thermal states of selected three-qubit spin chains, and noisy GHZ and W states.  We observe signatures of quantum phase transitions by quantifying entanglement in spin chains. We also demonstrate that entanglement between non-neighbouring sites can be activated by tuning the external magnetic field. In all tested cases, the gap between the lower and upper bounds is found to be very small, indicating that entcalc provides highly accurate estimates of the geometric entanglement for these states.
\end{abstract}

\maketitle

\section{Introduction}
Quantum entanglement is a fundamental form of correlation between quantum particles, enabling communication protocols that outperform their classical counterparts~\cite{Horodeckientanglementreview}. As a uniquely quantum phenomenon, entanglement serves as a key resource for tasks such as quantum teleportation~\cite{Benettteleportofstate}, superdense coding~\cite{Benettsuperdensecoding}, and quantum key distribution~\cite{BennettQKD}. Many of these applications require highly entangled states. Therefore, determining which quantum states possess greater amounts of entanglement is of central importance.  

To quantify the amount of entanglement in a given quantum state, several entanglement measures have been introduced. The most prominent examples include entanglement entropy~\cite{Nielsencambridgeup2010}, which is an entanglement measure only for pure states, negativity~\cite{Vidalneg}, and geometric entanglement~\cite{Weipra68_2003}. Each of these measures has its own advantages and limitations. The minimal requirement for any valid entanglement measure is that it should be non-increasing under local operations and classical communication (LOCC)~\cite{BenettLOCC}. One can further require that the entanglement measure is zero if and only if a state is separable. 
Unfortunately, quantifying entanglement in a given state is a difficult task. Moreover, even a simpler problem, deciding whether a given state is separable or entangled, is an NP-hard problem~\cite{gurvitssep,Gharibianse}.
Since determining separability itself is NP-hard, evaluating an entanglement measure for a general mixed state, which involves optimization over separable states,  is also computationally demanding. 

To address these issues, several simplifications and alternative approaches have been developed. For instance, in the qubit–qubit and qubit–qutrit cases, the negativity of a state is equal to zero if and only if the state is separable~\cite{Horodeckipptcriterion,Perespptcriterion}. In higher-dimensional systems, negativity can still detect some, but not all entangled states. Another class of tools are the so-called entanglement witnesses~\cite{Horodeckipptcriterion,Terhaldetecting}. A hermitian operator 
$W$ is called an entanglement witness if $\operatorname{Tr}(\rho W)\geq0$ for all separable states and
$\operatorname{Tr}(\rho W)<0$ implies that the state 
$\rho$ is entangled. Every entangled state has a witness that detects it. Therefore, entanglement witnesses allow one to certify the presence of entanglement, although one witness cannot detect every entangled state. For that reason, one needs to use many witnesses to increase efficiency of entanglement detection. 

An $n$-partite pure quantum state is fully separable if it can be written as~\cite{Horodeckientanglementreview}: 
\begin{equation}
    \ket{\psi}=\ket{\phi_{1}}\otimes\ldots \otimes\ket{\phi_{n}},
    \label{eq::fulsep}
\end{equation}
where $\ket{\phi_{j}}$ are pure states of subsystem $j$. Otherwise, it is called multipartite entangled. A mixed quantum state is multipartite entangled if it cannot be represented as a convex combination of fully separable pure states. Multipartite entanglement is a resource for various information-processing tasks~\cite{Leemul,Eppingmul}.
In many scenarios, we are interested in quantifying the amount of multipartite entanglement in a given state. It is also interesting to consider the type of the multipartite entanglement. For example, when considering entangled states shared between four parties $ABCD$, the states $\ket{\Phi}_{AB}\otimes\ket{\Phi}_{CD}$, where $\ket{\Phi}$ is a Bell state shared between $A$ and $B$, and $\ket{\text{GHZ}}_{ABCD}$ are quite different types of multipartite entanglement in terms of which of the four parties share entanglement. For this purpose, the notion of genuine multipartite entanglement was introduced. A quantum state is said to be genuinely multipartite entangled if it cannot be written as a convex combination of states that are separable with respect to any bipartition of the system~\cite{Svetlichnyge,Horodeckientanglementreview}. 

As mentioned earlier, computing the entanglement of a bipartite state is already a difficult problem. Multipartite entanglement exhibits an even richer and more complex structure: every bipartite entangled state is multipartite entangled, but some multipartite entangled states are separable across some bipartition. Consequently, evaluating the entanglement of a multipartite quantum state is an even more demanding task.

Due to these challenges, studies typically follow one of three approaches.
In the first approach, one restricts the analysis to pure states. A recent example in this direction is a method for estimating the geometric entanglement of pure multipartite states~\cite{Fitterestimatingentanglementrandommultipartite}.

The second approach reduces the complexity by considering only pairs of subsystems of the global state. For instance, in multiqubit systems, one may analyze only the entanglement of two qubit subsystems~\cite{Wanggme,Chenqft,Osborneqft}. An advantage of this method is that an analytical formula exists for the geometric entanglement of two-qubit states, making computations significantly more tractable~\cite{Weipra68_2003}. However, this simplification comes at the cost of ignoring the full multipartite entanglement structure. For example, the GHZ state~\cite{Greenbergerghzsta} is highly entangled, yet all of its two-qubit reduced states are separable.

Another approach is to focus exclusively on genuine multipartite entanglement (GME). A state is classified as GME if it cannot be expressed as a convex combination of biseparable states. This mirrors the definition of multipartite entanglement, where a state is considered multipartite entangled when it cannot be written as a convex combination of fully separable states. 
Several witnesses of genuine entanglement have been introduced ~\cite{Gaomulwitn,Bournamemulwit}. 
Recent research on GME focus on deriving computable lower bounds
~\cite{Magme,Jungnitschlower,Ligme,Daigme,Wanggme}.

In this work, we present a package to compute the geometric entanglement \cite{Weipra68_2003} of pure and mixed bipartite and multipartite quantum states. The novelty of our package lies in the combination of lower bounds with an upper bound on the geometric measure of entanglement, derived in~\cite{Streltsovupperbound}.
Simultaneous computation of both lower and upper bounds makes our method significantly more powerful than previous approaches, as it provides an interval within which the true value of entanglement must lie. For lower bounds on the geometric measure of entanglement for the bipartite scenario, we use the lower bounds originally derived in~\cite{Ganardilocalpurity}. As part of our theoretical contribution, we extend this analysis to derive novel lower bounds on the geometric measure of entanglement for the multipartite scenario.
Additionally, we employ two other lower bounds based on \cite{Streltsov_2010,Dohertyksymmetric}, to improve the range of tools available for estimating entanglement. Choosing between different bounds allows users to adjust computation precision and speed of computations. To the best of our knowledge, our package constitutes the first package allowing for high-precision computation of geometric entanglement for multipartite mixed states.

We provide two implementations of our package: \textbf{entcalc}~\cite{entcalc} for MATLAB and \textbf{entcalcpy}~\cite{entcalcpy} for Python. 
In MATLAB, we make use of the external toolboxes: \texttt{cvx}~\cite{cvx,Grantcvx} and \texttt{qetlab}~\cite{Johnstonqetlabpackagematlab}. 
In Python, our implementation relies on \texttt{NumPy}~\cite{Harrisnumpy}, \texttt{QuTiP}~\cite{JOHANSSONqutippython}, \texttt{CVXPY}~\cite{Diamond2016cvxpy}, and \texttt{SciPy}~\cite{Virtanenscipypython}.

Beyond its methodological novelty, our approach has practical implications for a wide range of problems in quantum information theory. 
We show that it can be directly applied to evaluate entanglement in bound entangled states (section~\ref{se::ppt}), to evaluate the entanglement of low rank states (section~\ref{se::low_rank}), to quantify entanglement in spin-chain systems (section~\ref{se::spinxx}), and to investigate the distribution and dynamics of entanglement in noisy quantum channels, (section~\ref{se::noise}). In many cases, the gap between the lower and upper bounds is found to be very small, confirming that the bounds we use are tight and allow for high-precision estimation of the geometric entanglement.
We believe that the presented package will serve as a valuable numerical tool for theoretical studies of entanglement in multipartite systems.

This paper is organized as follows. In Section~\ref{s::ge}, we introduce the geometric measure of entanglement. Section~\ref{s::a} provides a detailed description of our algorithm. We begin by outlining the procedure used to compute the entanglement of pure states, after which we present the algorithm for mixed states. In particular, we give explicit formulations of the four lower bounds employed in our method.
In Section~\ref{s::ap}, we apply our package to several representative examples. Section~\ref{s::compa} is devoted to a performance comparison of the lower bounds introduced in Section~\ref{s::a}. Finally, we conclude the paper in Section~\ref{s::con}.

\section{Geometric entanglement}{\label{s::ge}}
In this section, we formally define the geometric entanglement of a given state and discuss some relevant properties. 
For a pure quantum state $\ket{\psi}$, the geometric entanglement is defined as~\cite{Bihamgrov,Weipra68_2003}.
\begin{equation}
    E_G(\ket{\psi}) = 1 - \max_{\ket{\phi} \in S} \left| \bra{\phi} \ket{\psi} \right|^2,
    \label{Egpure}
\end{equation}
where $S$ denotes the set of fully separable states. For mixed states, the geometric entanglement is defined via the convex-roof extension~\cite{Weipra68_2003}:
\begin{equation}
    E_G(\rho) = \min_{\{p_i, \ket{\psi_i}\}} \sum_i p_i \, E_G(\ket{\psi_i}),
    \label{eq::ge}
\end{equation}
where the minimization is taken over all possible pure-state decompositions of $\rho = \sum_i p_i \ket{\psi_i}\bra{\psi_i}$.

There exists a useful relation connecting the geometric measure of entanglement with the quantum fidelity~\cite{Streltsov_2010}:
\begin{equation}
    E_G(\rho) = 1 - \max_{\sigma \in S} F(\rho, \sigma),
    \label{eq::gewithF}
\end{equation}
where the fidelity between two density matrices $\rho$ and $\sigma$ is defined as 
\begin{equation}
    F(\rho, \sigma) = \left( \mathrm{Tr} \sqrt{ \sqrt{\rho} \, \sigma \, \sqrt{\rho} } \right)^{2}.
\end{equation}

In general, the computation of geometric entanglement is a highly demanding numerical task. 
For two-qubit states, however, there exists a useful analytical relation between the geometric entanglement and concurrence \cite{Weipra68_2003}: 
\begin{equation}
    E_G(\rho) = \tfrac{1}{2} \left( 1 - \sqrt{1 - C(\rho)^2} \right),
    \label{eq::2qubge}
\end{equation}
where $C(\rho)$ denotes the concurrence~\cite{Hillc,Wootersform}. 
The concurrence is defined as
\begin{equation}
    C(\rho) = \max \big( 0, \lambda_1 - \lambda_2 - \lambda_3 - \lambda_4 \big),
\end{equation}
where $\lambda_i$ are the eigenvalues, in decreasing order, of the matrix 
$R = \sqrt{ \sqrt{\rho} \, \tilde{\rho} \, \sqrt{\rho} }$ with 
$\tilde{\rho} = (\sigma_y \otimes \sigma_y) \rho^* (\sigma_y \otimes \sigma_y)$, where $\sigma_y$ is the Pauli $Y$ matrix. Equation~\eqref{eq::2qubge} allows for easy computation of geometric entanglement in the two-qubit case. 

Geometric measure of entanglement can also be analytically computed for bipartite pure states~\cite{Weipra68_2003}, isotropic states~\cite{Weipra68_2003}, symmetric states~\cite{Weipra68_2003,Hubenerge}, symmetric Dicke states~\cite{Weipra68_2003,Chenge}, three qubit pure states~\cite{Tamaryange}. 
However, in the general multipartite scenario, no closed analytical expression is known and one must rely on computable lower and upper bounds. 
These bounds will be described in detail in the following section devoted to our algorithm.
\section{Algorithm}{\label{s::a}}
In this section, we present our algorithm for computing the geometric measure of entanglement. We begin by recalling how to estimate geometric entanglement for pure states. In this case, the problem can be formulated as a semi-definite program (SDP). Furthermore, by analyzing numerical solution of the SDP, we can find the accuracy of the computation. We then move on to the mixed-state scenario, where no direct accuracy certificate is available. To overcome this limitation, we discuss how to compute both upper and lower bounds on the geometric entanglement.

\subsection{Pure states}

In this section, we present two methods for computing the geometric entanglement of a pure state. The first method is based on an SDP formulation and is the one implemented in our package. The second method is a gradient-descent based approach. 

\subsubsection{Approach based on SDP}

An SDP based approach for computing the geometric entanglement of pure states presented here was introduced in~\cite{Ganardilocalpurity}. 
For completeness, we briefly recall it here. 
Here, we present the formulation for a tripartite state; 
the generalization to an arbitrary number of subsystems is straightforward.

The geometric entanglement of a pure tripartite state $\ket{\psi}^{ABC}$ satisfies the following lower bound~\cite{Ganardilocalpurity}:
\begin{equation}
    E_G(\ket{\psi}^{ABC}) \geq 1 - \max_{\sigma^{AB} \in \mathrm{PPT}} \Tr(\rho^{AB} \sigma^{AB}),
    \label{eq::lopu}
\end{equation}
where $\rho^{AB}$ denotes the reduced state obtained by tracing out the subsystem $C$, and $\sigma^{AB}$ is a valid quantum state, i.e., a positive semi-definite operator with unit trace. 

In~\cite[Proposition 2]{Ganardilocalpurity}, it was shown that the corresponding upper bound takes the form 
\begin{equation}
    E_G(\ket{\psi}^{ABC}) \leq 1 - \max_{\sigma^{AB} \in \mathrm{PPT}} \Tr(\rho^{AB} \sigma^{AB}) + 4\sqrt{\varepsilon},
    \label{eq::lopuep}
\end{equation}
where $\varepsilon = 1 - \lambda_{\max}$, and $\lambda_{\max}$ denotes the largest eigenvalue of the optimal $\sigma^{AB}$ 
that maximizes the expression above. Thus, $4\sqrt{\varepsilon}$ is a method accuracy.
It can be generalized to the multipartite scenario, see Appendix \ref{se::pu}. 

\subsubsection{Approach based on gradient descent}
 Upper bound algorithm used in this method was derived in~\cite{Shimoniub,Mostub}:
 Let $\ket{\psi}$ be an $n$-partite state for which we aim to estimate an upper bound on the geometric measure of entanglement. To do so, we need to find the closest separable state $\ket{\phi}.$
We begin with a random separable decomposition 
$\ket{\phi}=\ket{\phi_1}\ket{\phi_2}\cdots\ket{\phi_n}$. 
Next, we define 
$
\ket{\tilde{\psi}}
    = \bra{\phi_2}\bra{\phi_3}\cdots\bra{\phi_n}\ket{\psi},
$
which is an unnormalized pure state on the Hilbert space of the first subsystem. 
We then replace $\ket{\phi_1}$ by the normalized vector proportional to $\ket{\tilde{\psi}}$, which increases the overlap $\left|\braket{\phi|\psi}\right|$ 
The same update is performed sequentially for $i=2,3,\ldots,n$. 
This completes a single iteration of the algorithm. 
The procedure is repeated until the product state converges to a fixed point, up to a prescribed tolerance~$\epsilon$.

The algorithm presented above computes an upper bound on the geometric entanglement. It operates through an iterative procedure in which, at each step, we identify a separable state that is closer to the input state. However, the algorithm may converge to a separable state that is not the closest one. Consequently, this method does not guarantee that the true geometric entanglement is found. What it does guarantee is that the value obtained constitutes a valid upper bound.

This method is also used as a subroutine in the computation of the upper bound on the geometric entanglement of a mixed state. More specifically, when evaluating the upper bound for a mixed state, we iteratively refine the pure states appearing in its separable decomposition, and this refinement is performed using the aforementioned method.

\subsection{Mixed states}
In the case of mixed states, we estimate the geometric entanglement by calculating both lower and upper bounds. 
\subsubsection{Upper bound}
The algorithm for computing the upper bound was introduced in~\cite{Streltsovupperbound}. 
In contrast to the methods used to calculate lower bounds, the algorithm for the upper bound is based on gradient descent.

For mixed states the algorithm proceeds as follows. 
Consider a density operator $\rho$ of dimension $d$ for which we seek an upper bound on the geometric measure of entanglement. 
The task is to find a separable state $\sigma$ that maximizes Eq.~\eqref{eq::gewithF}. 
By Carathéodory's theorem, $\sigma$ requires at most $d^{2}$ pure states in its decomposition \cite{Nielsencambridgeup2010}. 
We begin by generating a random separable decomposition of~$\sigma$. 
The algorithm then iteratively refines this decomposition in order to maximize the fidelity between $\rho$ and $\sigma$. 
Each iteration consists of the following steps:
\begin{enumerate}
    \item  
    At the beginning of each iteration we have a decomposition of $\rho$ given by $\{p_i, \ket{\psi_i}\}$ and a decomposition of $\sigma$ given by $\{q_i, \ket{\phi_i}\}.$ We start by adjusting the decomposition of~$\rho$. We do it as follows: first, we find matrix $A=\sum_{i,j}\sqrt{p_iq_j}\braket{\phi_j|\psi_i}\ket{i}\bra{j}.$ Then, we compute singular value decomposition $A=VDW^{\dagger}.$ We define $U=W^{\dagger}V^{\dagger}.$ Then, we compute 
    \begin{equation}
        \ket{\alpha_i}=U_{ij}\sqrt{p_j}\ket{\psi_j},
    \end{equation}
    where $U_{ij}$ are matrix elements of $U$ in computational basis. Next, we compute $p_i'=\braket{\alpha_i|\alpha_i}$ and $\ket{\psi_i'}$ is normalized $\alpha_i.$
    After this step, we obtain a new ensemble $\{p_i', \ket{\psi_i'}\}$.

    \item  
    In this step we update $\sigma.$ 
    For each~$i$ we update $\ket{\phi_i}$ to a new state $\ket{\phi_i'}$ that is closer to $\ket{\psi_i'}$.  
    This is done using the same procedure as in the pure-state case.  
    The update is performed for all~$i$.

    \item  
    The coefficients $q_i$ are replaced by 
   $
        q_i' = 
        \frac{p_i' \, \left|\braket{\phi_i'|\psi_i'}\right|^{2}}
        {\sum_{k} p_k' \, \left|\braket{\phi_k'|\psi_k'}\right|^{2}}.
    $

    \item 
    We compute the fidelity between the updated $\sigma$ and~$\rho$.  
    If the increase in fidelity falls below a prescribed threshold, the algorithm terminates and we output an upper bound of the geometric entanglement.
\end{enumerate}

For more details and directions on how the algorithm is implemented, we refer to \cite{Streltsovupperbound}.

\subsubsection{Lower bounds}
Computation of lower bound on geometric entanglement is more complicated. We derive four different versions of the lower bound, each with its own advantages and limitations, see Section \ref{s::compa}.
One convenient method follows directly from Eq.~\eqref{eq::gewithF}. 
Instead of performing the maximization over the set of separable states, 
we relax the separability condition and maximize over the larger set of states that are positive under partial transposition (PPT). 
This leads to the following relation:
\begin{equation}
1 - \max_{\sigma \in \mathrm{S}} F(\rho, \sigma) \geq 1 - \max_{\sigma \in \mathrm{PPT}} F(\rho, \sigma),
\label{eq::gePPT}
\end{equation}
where $\mathrm{S}$ denotes the set of separable states and $\mathrm{PPT}$ the set of states with a positive partial transpose. 
In the qubit-qubit and qubit-qutrit cases, the inequality above becomes an equality~\cite{Horodeckipptcriterion,Perespptcriterion}. 

The root fidelity between two quantum states can be formulated as an SDP problem, which enables efficient numerical computation of the lower bound~\cite{Skrzypczaksdpi}. 
For clarity of presentation, in what follows we provide the formulation of the lower bound for bipartite systems, but it can be easily generalized for multipartite case, see Appendix \ref{se::multi}.
\begin{low}{\label{bo::pptf}}
    We can bound the geometric entanglement from below by the following SDP.
    \begin{equation}
        E_{g}(\rho^{AB})\geq1-\left(\max_{\sigma^{AB},\, X} \; \, \mathrm{Tr} \frac{1}{2}\left(X+X^{\dagger}\right)\right)^{2}
\label{eq::bopptf}
    \end{equation}
    Subject to the following constraints:
    \[
\begin{aligned}
& \sigma^{AB} \succeq 0, \quad 
 \mathrm{Tr}(\sigma^{AB}) = 1, \quad 
 \sigma^{AB,T_B} \succeq 0, \quad 
\\
&\begin{pmatrix}
\rho^{AB} & X \\
X^\dagger & \sigma^{AB}
\end{pmatrix}
\succeq 0. \quad 
\end{aligned}
\]
\end{low}
The lower bound~\eqref{eq::bopptf} can be efficiently computed and provides a reliable estimate of geometric entanglement for a wide range of states. 
However, it becomes inaccurate for certain classes of entangled states — in particular, it yields zero for PPT entangled states\cite{HORODECKIpptent}. 

The accuracy of the lower bound can be systematically improved by incorporating the hierarchy of $k$-symmetric extensions~\cite{Dohertyksymmetric}. 
A bipartite state $\rho_{AB_1}$ is said to have a $k$-symmetric extension if there exists $\rho_{AB_1\ldots B_k}$ satisfying:
\begin{equation}   
\Tr_{B_1\ldots B_{i-1}B_{i+1}B_{k}}{\rho^{AB_1\ldots B_k}}=\Tr_{B_2\ldots B_{k}}{\rho^{AB_1\ldots B_k}}
\label{eq::ks}
\end{equation}
 for all $i=1,2,\ldots ,k$. 
These additional constraints define a hierarchy of increasingly tighter relaxations, indexed by $k$, 
that converge to the exact set of separable states in the limit $k \rightarrow \infty$ \cite{Dohertyksymmetric}. 
In practice, however, increasing $k$ rapidly enlarges the dimension of the optimization problem, 
making computations infeasible for large $k$. 
Therefore, this approach provides a trade-off between computational efficiency and tightness of the lower bound. 
The $k$-symmetric extension condition can also be formulated within the SDP framework, allowing for practical implementation.
 \begin{low}{\label{bo::pptk}}
    We can bound the geometric entanglement from below by the following SDP~\cite{Philipfide}. 
    \begin{equation}
        E_{g}(\rho^{AB_1})\geq1-\left(\max_{\sigma^{AB_1}, X} \;  \, \mathrm{Tr} \frac{1}{2}\left(X+X^{\dagger}\right)\right)^{2}
\label{eq::bopptk}
    \end{equation}
    Subject to the following constraints:
    \[
\begin{aligned}
& \sigma^{AB_1B_2\ldots B_k} \succeq 0, \quad  
 \mathrm{Tr}(\sigma^{AB_1}) = 1, \quad \\
 &\sigma^{AB_1\ldots B_k,T_i} \succeq 0, \quad where \quad i=A,B_1,\ldots ,B_k \\
&\begin{pmatrix}
\rho^{AB} & X \\
X^\dagger & \sigma^{AB_1}
\end{pmatrix}
\succeq 0, \quad \\
& \Tr_{B_1\ldots B_{i-1}B_{i+1}B_{k}}{\sigma^{AB_1\ldots B_k}}=\Tr_{B_2\ldots B_{k}}{\sigma^{AB_1\ldots B_k}}
\end{aligned} 
\]
\end{low}

In Appendix \ref{se::k} we show how we construct constraint for k-symmetric extension in a multipartite case.

An alternative approach to estimating the lower bound of geometric entanglement arises from the purity–entanglement complementarity relation derived in~\cite{Ganardilocalpurity}. Purity-entanglement complementary relation provides connections between thermodynamics and entanglement theory. Assume that Alice helps Bob to cool down his system to a ground state. Both of them can do local operations, Alice can additionally send classical information to Bob. The maximal fidelity with ground state achieveable in this regime is given by: $F^{AB}_{\rightarrow}(\rho^{AB})$. For a detailed discussion of this and how $F^{AB}_{\rightarrow}(\rho^{AB})$ is related with entanglement, see \cite{Ganardilocalpurity}. Purity-entanglement complementary relation for the multipartite scenario can be stated as follows.
\begin{thm}\label{thm:pec_mult}
For any multipartite pure state $\psi^{AB_{1}\cdots B_{M}}$, the following equation holds:
    \begin{equation}
            F^{AB}_{\rightarrow}(\rho^{AB_1\ldots B_{M-1}}) + E_{g}(\rho^{B_1\ldots B_M}) = 1.
            \label{eq::thm}
\end{equation}
\end{thm}
\begin{proof}
    See Appendix \ref{se::pr} for proof.
\end{proof}
 
We can upper bound $F^{AB}_{\rightarrow}(\rho^{AB_1\ldots B_{M-1}})$ by:
\begin{equation}
    F^{AB}_{\rightarrow}(\rho^{AB_1\ldots B_{M-1}})\leq \max_{X^{AB_1\ldots B_{M-1}}}\Tr(X^{AB_1\ldots B_{M-1}}\rho^{AB_1\ldots B_{M-1}}),
\end{equation}
where
\begin{equation}
        X^{AB_{1}\cdots B_{M-1}}\geq 0,\quad X^{A} = \openone^{A},\quad X^{T_{\tilde{B}}}\geq 0\,\, \forall \tilde{B},
\end{equation}
where $\tilde{B}$ is any bipartition of $B_1\ldots B_{M-1}.$
Using this relation, we can formulate the lower bound: 
 \begin{low}{\label{bo::purc}}
    We can bound the geometric entanglement from below by the following SDP. 
    \begin{equation}
        E_{g}(\rho^{BC})\geq1-\max_{X} \;  \, \mathrm{Tr} \left( X^{AB}\rho^{AB}\right)
        \label{eq::bopurc}
    \end{equation}
    Subject to the following constraints:
    \[
\begin{aligned}
& X^{AB} \succeq 0 \quad & 
& X^{A} = \openone^{A} \quad 
& X^{AB,T_B} \succeq 0 \quad 
& 
\end{aligned} 
\]
\end{low}
In the expressions above, $\rho^{ABC}$ denotes a purification of the mixed state $\rho^{BC}$, 
while $\rho^{AB}$ represents the reduced density matrix obtained by tracing out subsystem $C$. 

We can also state one more bound~\cite{Ganardilocalpurity}. It is based on the same relation as bound \eqref{eq::bopurc}. The difference is that we operate on the whole purified system. Consequently, we operate on higher dimensional matrices, which makes computations longer, but it also gives tighter bound. 
 \begin{low}{\label{bo::purs}}
    We can bound the geometric entanglement from below by the following SDP.
    \begin{equation}
        E_{g}(\rho^{BC})\geq1-\max_{X} \;  \, \mathrm{Tr} \left( X^{ABC}\rho^{ABC}\right)
        \label{eq::bopurs}
    \end{equation}
    Subject to the following constraints:
    \[
\begin{aligned}
& X^{ABC} \succeq 0 \quad 
 X^{A} = \openone^{A} \quad 
X^{ABC,T_i} \succeq 0 \quad 
\textit{where i=A,B,C} 
\end{aligned} 
\]
\end{low}

In the multipartite scenario, the constraints are largely analogous to the bipartite case. 
The only modification is that the auxiliary variable $X$ is required to be PPT with respect to relevant bipartitions. 
A formal proof of this statement is provided in the Appendix \ref{se::pr}. In Appendix \ref{se::multi} we show how the bounds look in multipartite scenario.

\section{Applications}\label{s::ap}
In this section, we present results obtained using our package to compute entanglement in several quantum states. In all cases, we calculated the lower bound using the bound \eqref{eq::bopurs}. Whenever we state that we computed geometric entanglement, we mean multipartite geometric entanglement, defined by the Eq. \eqref{eq::ge}, unless we specify otherwise.

\subsection{PPT entangled state}{\label{se::ppt}}
We begin by calculating the entanglement in the family of $3 \otimes 3$ Horodecki PPT entangled states~\cite{HORODECKIpptent}. 
These states are parameterized by $a\in [0,1]$ and are defined as follows:
\begin{equation}
    \rho_a = \frac{1}{8a + 1}
\begin{bmatrix}
a & 0 & 0 & 0 & a & 0 & 0 & 0 & a \\
0 & a & 0 & 0 & 0 & 0 & 0 & 0 & 0 \\
0 & 0 & a & 0 & 0 & 0 & 0 & 0 & 0 \\
0 & 0 & 0 & a & 0 & 0 & 0 & 0 & 0 \\
a & 0 & 0 & 0 & a & 0 & 0 & 0 & a \\
0 & 0 & 0 & 0 & 0 & a & 0 & 0 & 0 \\
0 & 0 & 0 & 0 & 0 & 0 & \tfrac{1+a}{2} & 0 & \tfrac{\sqrt{1 - a^2}}{2} \\
0 & 0 & 0 & 0 & 0 & 0 & 0 & a & 0 \\
a & 0 & 0 & 0 & a & 0 & \tfrac{\sqrt{1 - a^2}}{2} & 0 & \tfrac{1+a}{2}
\end{bmatrix}
\label{eq::ppt33}
\end{equation}
We plot the lower bound on the geometric entanglement of this state as a function of the parameter $a$ in Fig.~\ref{fig::pptent}. 
To create this plot, the parameter $a$ was sampled with a step size of $0.01$, starting from $a = 0.01$. 
The maximum difference between the lower and upper bounds is $5.15\times10^{-6}$. 
Therefore, our method performs well in this case. 

\begin{figure}
    \centering
    \includegraphics[width=0.9\linewidth]{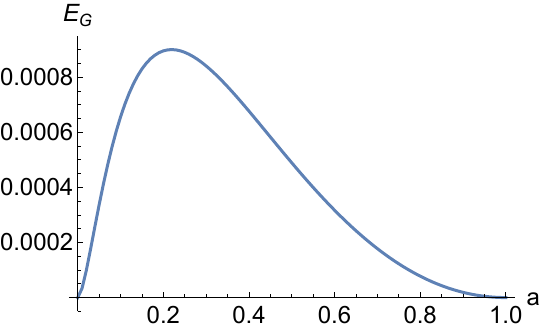}
    \caption{The lower bound on the geometric entanglement for the PPT entangled state defined in Eq.~\eqref{eq::ppt33}.
    Although the state remains PPT entangled for all $a\in(0,1),$ our method successfully quantifies its entanglement.
    The maximal difference between the lower and upper bounds is $5.15\times10^{-6}.$ }
    \label{fig::pptent}
\end{figure}

\subsection{Low-rank states}\label{se::low_rank}

Now consider the following convex combination of the GHZ and $W$ states:
\begin{equation}
    \rho(p) = p\,\ket{GHZ}\bra{GHZ} + (1 - p)\,\ket{W}\bra{W},
    \label{eq::ghzmix}
\end{equation}
where 

\begin{equation}
   \ket{GHZ}_N = \frac{1}{\sqrt{2}}\left(\ket{0}^{\otimes N} + \ket{1}^{\otimes N}\right),
   \label{eq::ghzn}
\end{equation}

\begin{equation}
    \ket{W}_N = \frac{1}{\sqrt{N}}\left(\ket{100\ldots 0} + \ket{010\ldots 0} + \ldots + \ket{0\ldots 01}\right).
    \label{eq::wn}
\end{equation}

\begin{figure}
    \centering
    \includegraphics[width=0.9\linewidth]{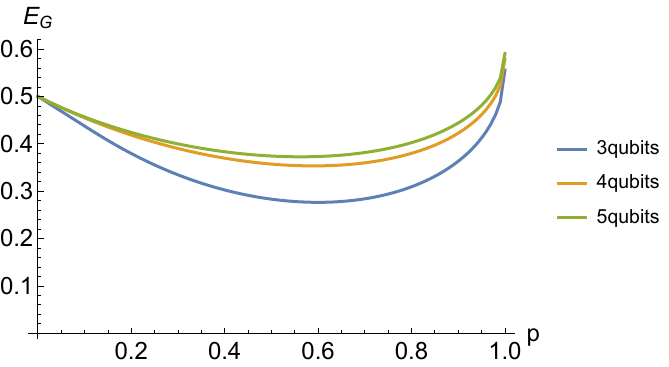}
    \caption{Plot of the geometric entanglement of the mixture from Eq. \eqref{eq::ghzmix} for a system consisting of 3, 4 and 5-parties, each having one qubit. The maximal difference between lower and upper bound was smaller than $5.05\times 10^{-6}$ for 3 parties, $1.5\times 10^{-5}$ for 4 parties and $3.5\times 10^{-3}$ for 5 parties.}
    \label{fig::ghw}
\end{figure}
The entanglement properties of this class of states were previously investigated in~\cite{Zhumix}, where the authors employed the robustness of entanglement~\cite{Vidalrobu} as the entanglement measure. They derived both lower and upper bounds on the robustness and showed that these bounds are nearly identical for three- and four-parties, each having one qubit. Here, we produce a similar analysis using the geometric measure of entanglement and extend the study to the five-qubit case. Figure~\ref{fig::ghw} presents the geometric entanglement of the state in~\eqref{eq::ghzmix} as a function of $p$. We plot it for three, four, and five parties, each having a single qubit.

\subsection{Thermal entanglement}{\label{se::spinxx}}
The study of thermal entanglement originated from investigations of entanglement in linear arrays of qubits~\cite{woottersthermal,O'Connorprathermal,Briegelthermal}. 
In these works, entanglement was analyzed both for ground states~\cite{O'Connorprathermal} and for systems at non-zero temperature~\cite{Arnesenprspinchains}. 
Consider a system described by a Hamiltonian $H$. Its state at thermal equilibrium at temperature $T$ is given by:
\begin{equation}
    \rho_T = \frac{e^{-\beta H}}{Z},
    \label{eq::thermal}
\end{equation}
where $\beta = 1/T$ is the inverse temperature (we set the Boltzmann constant $k_B = 1$) and $Z = \mathrm{Tr}(e^{-\beta H})$ is the partition function. 
The entanglement in the state defined by Eq.~\eqref{eq::thermal} is referred to as thermal entanglement~\cite{Arnesenprspinchains}. 

Up to now, most studies have focused on pairwise entanglement~\cite{Wangspins}, employed entanglement witnesses~\cite{Bruknerwitness}, or derived analytical results on multipartite entanglement for specific models~\cite{Stelmachovicthermal}. 
Some lower bounds on multipartite entanglement in spin systems have been obtained~\cite{Gühnespinchains}. 
For a comprehensive review of entanglement in spin chains, we refer the reader to ~\cite{Amicothermalreviewentang}. 

In this section, we demonstrate that our package can be effectively applied to extend these studies and provide an efficient method for the quantitative analysis of thermal multipartite entanglement.

Consider a one-dimensional lattice of spin-$\frac{1}{2}$ particles that can interact with one another. 
At first glance, this appears to be a simple model; however, it can effectively capture interaction mechanisms present in more complex systems~\cite{Breunigspinchains}. 
If we assume that each spin interacts only with its nearest neighbors, the Hamiltonian of the system takes the form:
\begin{equation}
\begin{aligned}
H = &\ \frac{1}{2}\left(\sum_j\left(J_x\sigma_j^{x}\sigma_{j+1}^{x}
     + J_y\sigma_j^{y}\sigma_{j+1}^{y}
     + J_z\sigma_j^{z}\sigma_{j+1}^{z}\right)\right) \\
    &\ + h\sum_j\sigma_j^{z},
\end{aligned}
\label{genhamod}
\end{equation}
 where $h$ denotes the strength of the external magnetic field, $J_i$ is the coupling constant, 
$\sigma_j^{i} = \openone^{\otimes (j-1)} \otimes \sigma^{i} \otimes \openone^{\otimes (n-j)}$ is a Pauli matrix acting on a subsystem $j$ and $i=0,1,2,3$ numerates the Pauli operator. 
Here, $n$ represents the total number of spins. 
In this work, we impose periodic boundary conditions, namely $\sigma_{n+1} = \sigma_{1}.$ 
The coupling constant $J_i$ characterizes the interaction strength between neighboring spins, 
and different spin models correspond to specific choices of $J_i.$  

In this paper, we analyze two values of the coupling constant, $J=\pm 1.$ We restrict ourselves to these cases because in the Hamiltonian \eqref{genhamod} all coupling parameters $J$ are assumed either to be equal or zero. Hence, a common factor 
$J$ can be extracted from the sum. As a result, all states under consideration take the form $e^{\, \beta J H'},$ where $H'$
denotes the Hamiltonian \eqref{genhamod} divided by the common factor 
$J.$ Consequently, the properties of the resulting state depend only on the product 
$\beta J.$ Without loss of generality, we fix 
$\abs{J}=1$ and vary the state by changing the sign of $\beta.$

 \subsubsection{Heisenberg XX model}

Before discussing the Heisenberg XX model, we would like to clarify the difference between pairwise and bipartite entanglement. Consider a three-qubit state $\rho^{ABC}.$ In this section, we usually compute multipartite entanglement of $\rho^{ABC},$ i.e. entanglement with respect to partition $2\otimes2\otimes2.$ If we trace out one of the subsystem, for example, $C$ and compute the entanglement of the remaining state $\rho^{AB}$, we obtain pairwise entanglement. In contrast, if we group two subsystems together and compute the entanglement of the resulting 
$4\otimes2$ bipartite state, we obtain bipartite entanglement. For instance, by joining subsystems 
$A$ and $B$, we obtain the bipartite state $\rho^{AB|C}$. 

In this section, we consider the case without an external magnetic field and focus on the XX model. 
The corresponding Hamiltonian, obtained from Eq.~\eqref{genhamod}, is given by
\begin{equation}
    H_{XX} = -\frac{J}{2} \sum_j \left( \sigma_{j}^{x}\sigma_{j+1}^{x} + \sigma_{j}^{y}\sigma_{j+1}^{y} \right).
    \label{eq::hxx}
\end{equation}
This type of Hamiltonian was previously examined in ~\cite{Wangspins}, where the authors studied the Heisenberg XX model for a three-qubit system. 
However, their analysis was restricted to pairwise entanglement, computed via Eq.~\eqref{eq::2qubge}, and did not address tripartite entanglement. 
They showed that for \( J < 0 \) no pairwise entanglement is present in the spin chain, while for \( J > 0 \) entanglement appears below the temperature \( T < 1.217 \). 

In this work, we revisit the three-qubit system from ~\cite{Wangspins} and compute the multipartite entanglement using our method. 
Figure~\ref{fig::hxx} shows the thermal entanglement of the Hamiltonian~\eqref{eq::hxx} for \( J = 1 \) and \( J = -1 \). 
As one can see from Figure~\ref{fig::hxx}, multipartite entanglement is present for both \( J > 0 \) and \( J < 0 \), even though pairwise entanglement vanishes for \( J < 0 \)~\cite{Wangspins}. The above result confirms that checking only pairwise entanglement is not sufficient.

Another interesting phenomenon appears in the parameter region just after entanglement first emerges for both $J\pm 1$. The lower bound, computed by our package, on the geometric entanglement becomes non-zero at approximately $\beta\approx0.68$ for $J=1$ and $\beta\approx0.65$ for $J=-1$. 
However, a straightforward analytical computation shows that every thermal state corresponding to the Hamiltonian~\eqref{eq::hxx} with coupling 
$J=-1$ remains PPT across all bipartitions for 
\( \beta < 0.781 \). (To find out that the state is PPT in a given bipartition, we computed negativity with respect to this bipartition.)
Similarly, all thermal states with Hamiltonian~\eqref{eq::hxx} and coupling 
$J=+1$ are PPT across all bipartitions for \( \beta  < 0.706 \). 
Consequently, every thermal state of the Hamiltonian \eqref{eq::hxx} with $J=+1$ and $\beta\in[0.65,0.706]$ is both entangled and PPT with respect to all bipartitions. Such states are bound entangled. Similarly, every thermal state with Hamiltonian \eqref{eq::hxx} and $\beta\in[0.65,0.781]$ exhibit the same behavior and are also bound entangled. In Fig.~\ref{fig::pptspi}, we plot the geometric entanglement of a XX model as function of $\beta$ but we include only those values of $\beta$ which gives PPT entangled state.
The existence of bound entanglement in spin chains was proven in ~\cite{Tothboundentanglementinspinchains}. 
In our work, we not only confirm the presence of bound entanglement in spin chains but also provide a quantitative characterization of it.
 
 \begin{figure}
     \centering
     \includegraphics[width=0.9\linewidth]{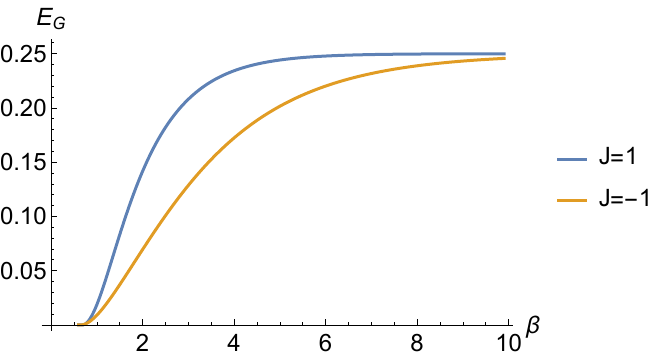}
     \caption{Thermal entanglement of three qubit Heisenberg spin chain with Hamiltonian \eqref{eq::hxx}, with $J=1$ in blue and $J=-1$ in orange. The difference between the calculated lower and upper bound difference is at most $2.42\times10^{-6}.$ State with $J=-1$ is multipartite entangled even though it is pairwise separable. }
     \label{fig::hxx}
 \end{figure}
\begin{figure}
     \centering
     \includegraphics[width=0.9\linewidth]{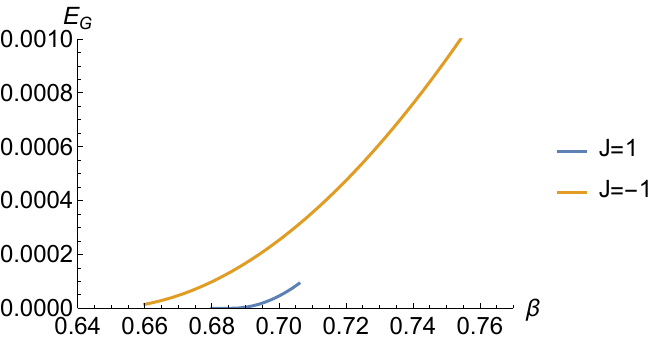}
     \caption{The geometric entanglement of a XX model. The plot has the same setting as Fig. \ref{fig::hxx}, but we plot only those values of $\beta$ which gives PPT entangled state. For $J=-1$ states with higher $\beta$ are bound entangled. The difference between the calculated lower and upper bound is atmost $4.5\times 10^{-6}$.}
     \label{fig::pptspi}
 \end{figure}

\subsubsection{XXX model with a magnetic field.}
Now, consider an XXX model with a magnetic field. The Hamiltonian \eqref{genhamod} is given by:
\begin{equation}
\begin{aligned}
H_{XXX} = &\ -\frac{J}{2}\sum_j\bigl(\sigma_{j}^{x}\sigma_{j+1}^{x}
    + \sigma_{j}^{y}\sigma_{j+1}^{y}
    + \sigma_{j}^{z}\sigma_{j+1}^{z}\bigr) \\
&\ + h\sum_{j}\sigma_{j}^{z}.
\end{aligned}
\label{eq::xxx}
\end{equation}
In this section we put $J=-1.$
Simply analytical computations show that the ground state corresponding to the thermal state (ground state is a state with $\beta\rightarrow\infty$) obtained from this Hamiltonian is highly entangled for a weak magnetic field (\( h = 1 \)), moderately entangled for an intermediate magnetic field (\( h = 1.5 \)), and separable for a strong magnetic field (\( h = 2 \)), see Appendix \ref{s::pha}. 
This result suggests the presence of a quantum phase transition~\cite{Sachdevqft} as the magnetic field strength changes. 
Indeed, by finding an analytical expression for the ground state, we observe a discontinuity at \( h = 1.5 \). 
Figure~\ref{fig::xxxma} shows the geometric entanglement of a thermal state \eqref{eq::xxx} as a function of the inverse temperature for these three values of the magnetic field.
To make this plot we sampled the inverse temperature from \( 0.3 \) to \( 6 \) with a step of \( 0.05 \), and from \( 6 \) to \( 10 \) with a step of \( 0.1 \). 

We also did more precise computations in the region where the entanglement emerges. Our lower bound becomes non-zero for \( \beta  \approx 0.365 \), independently of the magnetic field strength. 
Analytical computations show that the states become NPT in all bipartitions for \( \beta  > 0.462 \). 
Thus, there exists a range of \( \beta \) for which the thermal state is bound entangled also in this scenario.

We can also compute the entanglement as a function of the magnetic field. 
The phase transition becomes clearly visible at high values of the inverse temperature, as shown in Fig.~\ref{fig::mag}. 
For low temperatures ($\beta=10$), we observe a rapid change of the geometric entanglement around \( h = -1.5 \) and around \( h = 1.5 \). 
A similar change is also observed near \( h = 0 \). 
This behavior is consistent with analytical results, which indicate that in the limit \( \beta \rightarrow \infty \) there is a discontinuous change of the state at \( h = \pm1.5 \) and a single discontinuity at \( h = 0 \). 
At higher temperatures, this effect disappears. 

Entanglement in points close to quantum phase transitions has been studied extensively using concurrence~\cite{Chenqft,Osborneqft,Osterlohqft,Guqft,Glaserqft}. 
In ~\cite{Guqft,Glaserqft}, the authors showed that concurrence can exhibit discontinuities at quantum phase transitions. 
The works cited above focus on pairwise entanglement, whereas our method extends this analysis to the multipartite scenario.
\begin{figure}
    \centering
    \includegraphics[width=0.9\linewidth]{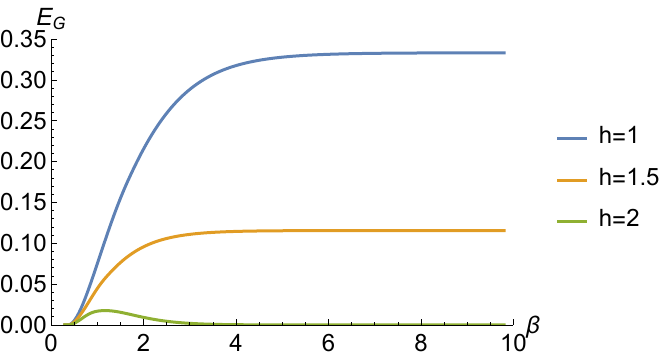}
    \caption{The plot shows the geometric entanglement of a spin chain described by the XXX model with a magnetic field, as a function of the inverse temperature. 
    Blue dots correspond to a weak magnetic field ($h = 1$), orange dots to a medium field ($h = 1.5$), and green dots to a strong field ($h = 2$). 
    Note the distinct values of geometric entanglement in the ground state ($\beta \rightarrow \infty$), which suggests the possible presence of a quantum phase transition in this system. 
    Each curve consists of 155 sample points, and the maximum estimation error is $3.81 \times 10^{-6}$. }
    \label{fig::xxxma}
\end{figure}
\begin{figure}
    \centering
    \includegraphics[width=0.9\linewidth]{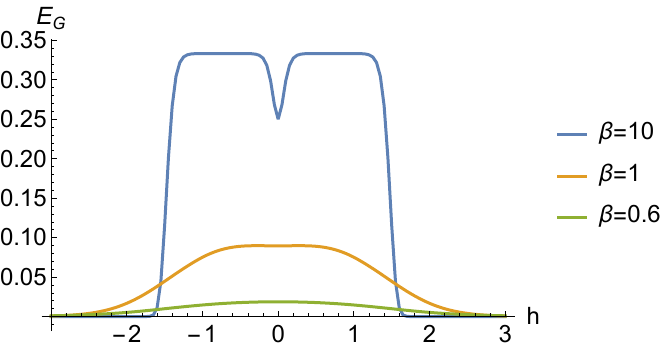}
    \caption{This figure depicts the geometric entanglement of a thermal state of XXX model as a function of the magnetic field. 
    Three configurations are shown: blue dots correspond to the low-temperature state ($\beta = 10$), orange dots to the medium-temperature state ($\beta = 1$), and green dots to the high-temperature case ($\beta = 0.6$). 
    The maximum difference between the lower and upper bounds on this plot is $4.45 \times 10^{-6}$. 
    For the low-temperature state, one can clearly observe a rapid change in entanglement near $h = -1.5$ and $h = 1.5$, which correspond to the points of quantum phase transition.}
    \label{fig::mag}
\end{figure}
\subsubsection{Hexagonal spin chain}
Assume that the spin chain has a length of 6 and periodic boundary conditions. 
The interactions are described by the XX model with an external magnetic field, and the Hamiltonian reads:
\begin{equation}
    H = -\frac{J}{2} \sum_j \left( \sigma_{j}^{x}\sigma_{j+1}^{x} + \sigma_{j}^{y}\sigma_{j+1}^{y} \right) 
        + h \sum_j \sigma_{j}^{z}.
    \label{eq::hxxwmf}
\end{equation}
We consider a three-qubit subsystem consisting of spins that do not interact directly with each other. 
There are two such possible subsystems, which are indicated by blue and green colors in Fig.~\ref{fig::hex}.
\begin{figure}
    \centering
    \includegraphics[width=0.5\linewidth]{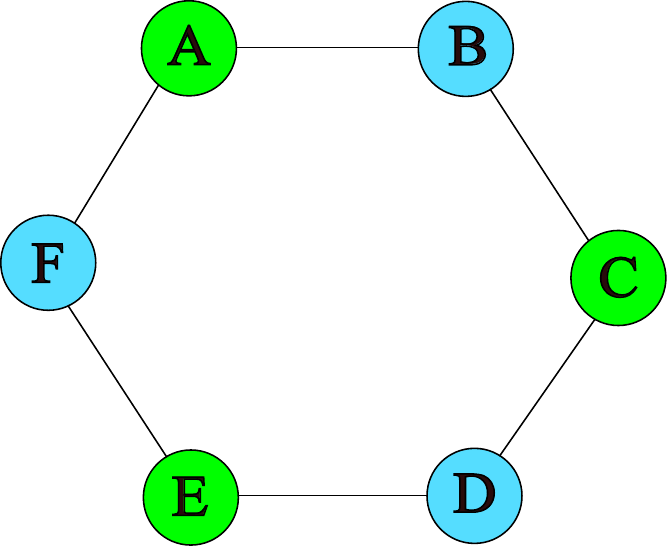}
    \caption{Six-qubit spin chain arranged in a circular structure. 
    Each qubit is represented by a circle and labeled with a letter. 
    Interactions between qubits are indicated by black lines. 
    Two distinct three-qubit subsystems are composed of qubits that do not interact directly with each other; 
    these subsystems are distinguished by different circle colors in the figure. }
    \label{fig::hex}
\end{figure}
Assume that the system is in a low-temperature environment with \( \beta = 5 \). 
We focus on the three-qubit subsystem \( B:D:F \), whose particles do not interact directly via this Hamiltonian. 
Interestingly, in this scenario, the application of a magnetic field can induce entanglement between non-interacting spins. 
As shown in Fig.~\ref{fig::hexag}, for both weak and strong magnetic fields, no thermal entanglement is observed. 
However, an intermediate magnetic field generates entanglement between non-neighboring sites of the chain. 
Figure~\ref{fig::hexag} was produced for the inverse temperature \( \beta = 5 \).

The three-partite entanglement emerges at \( h = 0.22 \). 
We also computed pairwise entanglement of the two-qubit subsystems \( B:D \), \( B:F \), and \( D:F \) using \eqref{eq::2qubge}. Our results show that none of these subsystems exhibit pairwise entanglement for \( h < 0.35 \), see Fig. \ref{fig::hexag}.
This observation further illustrates that in spin chains the entanglement structure is significantly richer than pairwise correlations alone. 
To obtain a complete understanding of the system, it is therefore necessary to employ entanglement measures on whole multipartite system.
\begin{figure}
    \centering
    \includegraphics[width=0.9\linewidth]{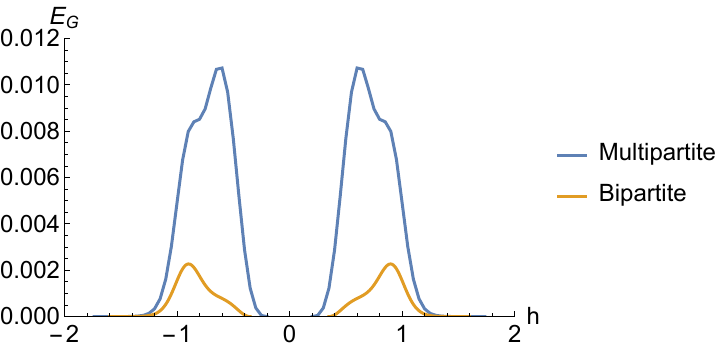}
    \caption{Geometric entanglement of the subsystem $B:D:F$ (blue plot) of the six-qubit thermal state governed by the Hamiltonian~\eqref{eq::hxxwmf}, plotted as a function of the magnetic field~$h$. 
    For $h=0$ there is no entanglement in the subsystem. 
    As $h$ increases, long-distance entanglement between non-interacting sites emerges, 
    but for strong magnetic fields ($h > 1.6$) the entanglement vanishes. 
    The maximum lower-upper bound difference in this plot is $6.6\times10^{-5}$. The orange plot is a plot of bipartite geometric entanglement between systems $D:F.$
    The plot was generated for the inverse temperature $\beta=5$. }
    \label{fig::hexag}
\end{figure}

\subsection{Noise in multipartite systems}{\label{se::noise}}
In this section, we study the impact of noise on highly entangled multipartite states, focusing on local noise. To analyze noise action we use quantum channel formalism and Kraus representation \cite{Nielsencambridgeup2010} of the quantum channel. 
We consider two types of noisy channels: amplitude damping and depolarizing channels. 
The Kraus operators for the amplitude damping channel are given by:
    \begin{equation}
    \begin{split}
        &E_{1}=
    \begin{bmatrix}
    1 && 0\\
    0 && \sqrt{1-q}
    \end{bmatrix},\\
    &E_{2}=\begin{bmatrix}
    0 && \sqrt{q}\\
    0 && 0
    \end{bmatrix}.\\  
    \end{split}
    \label{eq::gad}
\end{equation}
    Here, \( q \) is the damping parameter, satisfying \( 0 \leq q \leq 1 \). 
    
    Let us denote amplitude damping channel with parameter $q$, by $\Lambda_q^{AD}.$ 
We compare the effect of amplitude damping channel on the GHZ and W states for a three qubit system. To do so, we compute the multipartite geometric entanglement of the states:
\begin{equation}
    \Lambda_q^{AD}\otimes\Lambda_q^{AD}\otimes\Lambda_q^{AD}(\ketbra{GHZ})
    \label{eq::adg}
\end{equation}
and
\begin{equation}
    \Lambda_q^{AD}\otimes\Lambda_q^{AD}\otimes\Lambda_q^{AD}(\ketbra{W}).
    \label{eq::adw}
\end{equation}
Similar computations were performed in~\cite{Carvalhoamplitude,Alinoi}. 
In~\cite{Alinoi}, the authors employed a lower bound on genuine multipartite entanglement developed in~\cite{Jungnitschlower}. 
Their conclusion was that the W state is more robust against amplitude damping noise than the GHZ state. 
Our computations confirm this result. 
However, as shown in Fig.~\ref{fig::ghzm}, for low damping parameters the geometric entanglement of the GHZ state is initially larger than that of the W state. 
For higher noise parameters, the geometric entanglement of a W state becomes larger. 

We perform similar computations for depolarizing channel. Kraus operators of the depolarizing channel are:
\begin{align}
E_0 &= \sqrt{1-\frac{3p}{4}}\,I, \label{eq::E0}\\
E_1 &= \sqrt{\frac{p}{4}}\,X, \label{eq::E1}\\
E_2 &= \sqrt{\frac{p}{4}}\,Y, \label{eq::E2}\\
E_3 &= \sqrt{\frac{p}{4}}\,Z, \label{eq::E3}
\end{align}
where $X,Y,Z$ are the single qubit Pauli matrices. Let us denote depolarizing channel with strength $p$ by $\Lambda_{p}^{dep}.$ We compute the geometric entanglement of a state

\begin{equation}
    \Lambda_{p}^{dep}\otimes\Lambda_{p}^{dep}\otimes\Lambda_{p}^{dep}(\ketbra{GHZ})
    \label{eq::deg}
\end{equation}
and
\begin{equation}
    \Lambda_{p}^{dep}\otimes\Lambda_{p}^{dep}\otimes\Lambda_{p}^{dep}(\ketbra{W}).
    \label{eq::dew}
\end{equation}
Fig.~\ref{fig::dep} shows the entanglement of states~\eqref{eq::deg} and~\eqref{eq::dew} as a function of the depolarizing parameter. Curves for the geometric entanglement of both states look nearly identical, for $p>0.1$. Previous studies have performed similar comparisons between GHZ and W states. In~\cite{Alinoi} a lower bound on the genuine multipartite robustness of entanglement was used. This approach led to different conclusions: the robustness bound indicated that the GHZ state is more resilient to noise, while in our computations, both states performed similarly.
\begin{figure}
    \centering
    \includegraphics[width=0.9\linewidth]{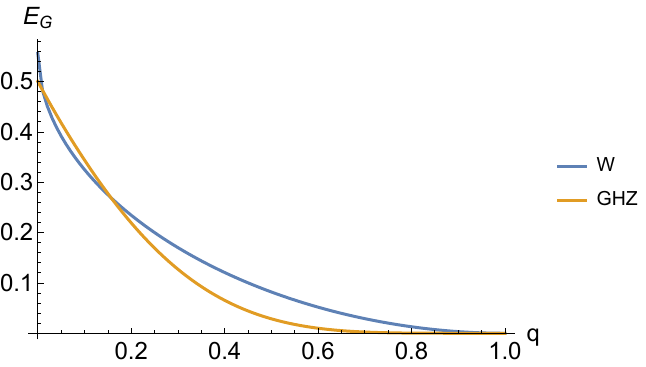}
    \caption{Geometric entanglement as a function of the damping parameter for the amplitude damping channel acting on GHZ and W states. 
    Blue dots represent the GHZ state, while orange dots correspond to the W state. 
    For small values of the damping parameter, the GHZ state exhibits higher entanglement; 
    however, as the noise strength increases, the W state becomes more robust and retains more entanglement. For W state difference between lower and upper bound was at most $4.5\times 10^{-8}.$ For GHZ state it was $4\times 10^{-3}.$ }
    \label{fig::ghzm}
\end{figure}
\begin{figure}
    \centering
    \includegraphics[width=0.9\linewidth]{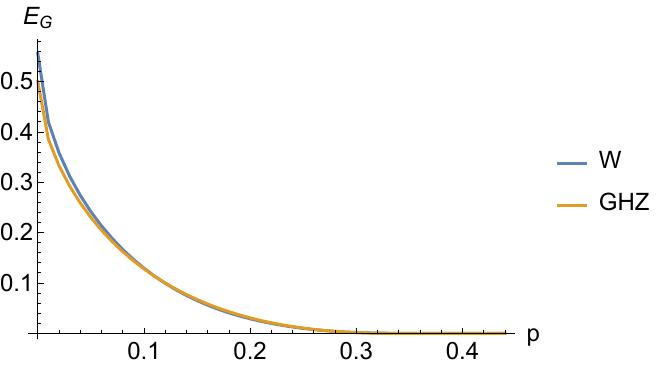}
    \caption{Plot of the geometric entanglement as a function of depolarizing parameter for the depolarizing channel. Blue dots represent GHZ state, while orange dots represent W state. In this scenario, entanglement is almost the same in both states. Highest lower-upper difference is $3\times 10^{-6}.$}
    \label{fig::dep}
\end{figure}
\section{Comparison of lower bounds}{\label{s::compa}}
To compare the performance of lower bounds on bipartite systems, we generated 100 random \(3 \times 3\) mixed states and computed the corresponding lower bounds for each. 
The mean computation times for each bound were as follows: 
for bound~\eqref{eq::bopptf}, \(1.008\ \mathrm{s}\); 
for bound~\eqref{eq::bopptk} with \( k = 2 \), \(3.03\ \mathrm{s}\) and with \( k = 3 \), \(228.2\ \mathrm{s}\); 
for lower bound~\eqref{eq::bopurc}, \(2.50\ \mathrm{s}\); 
and for bound~\eqref{eq::bopurs}, \(180\ \mathrm{s}\). Computations were done on a standard laptop (16GB RAM), using MATLAB version, all optional arguments of the functions had default values. 

The best results were obtained with bounds~\eqref{eq::bopurs} and \eqref{eq::bopptk} with \( k = 3 \), producing very similar outcomes. 
We performed computations with a precision of \(3 \times 10^{-8}\). 
In 60 cases, the difference between the bounds was smaller than the precision; in 98 cases it was below \(10^{-7}\); and in all remaining cases it was less than \(2 \times 10^{-6}\). 
In 96 instances, the solution given by bound~\eqref{eq::bopurs} yielded slightly better results. 
Although these two bounds give very similar results, bound~\eqref{eq::bopurs} is faster, making it preferable in practice, see Table \ref{tab::compa}. 

Sometimes, we are more interested in speed than in accuracy. In such cases, we can use the bounds~\eqref{eq::bopptk} with \( k = 2 \) and~\eqref{eq::bopurc}. They offer similar computation times and comparable accuracy, with the largest observed difference being \(8 \times 10^{-4}\). 
For 65 cases, bound~\eqref{eq::bopptk} with \( k = 2 \) was more accurate. 

Finally, bound~\eqref{eq::bopptf} provides the worst accuracy but is much faster than the other bounds, making it suitable for high-speed computations. 
However, it can only detect states with positive negativity and therefore cannot detect PPT-entangled states.

 In Appendix \ref{s::accura} we discuss the numerical precision of the package. We also provide an explanation for why, in some cases, the computed lower bound may slightly exceed the corresponding upper bound. 

\begin{table*}[t]
    \centering
    \begin{tabular}{c|c|c}
        Bound & Number of times given bound was most accurate & Average computation time in seconds\\
        \hline
         \eqref{eq::bopptk} & 0 & 1.008 \\
         \eqref{eq::bopptk} with $k=2$ & 0 & 3.03\\
         \eqref{eq::bopptk} with $k=3$ & 64 & 228.2\\
         \eqref{eq::bopurc} & 0 & 2.5\\
         \eqref{eq::bopurs} & 96 & 180
    \end{tabular}
    \caption{Comparison of the performance of lower bound. We sampled 100 random $3\times3$ mixed states. The total from the second column is greater than the number of total samples, because sometimes two bounds gave most accurate result up to computer precision. The best bound for the high accuracy precision is bound \eqref{eq::bopurs}. For a trade-off between speed and accuracy, we recommend bound \eqref{eq::bopptk} with $k=2,$ see comment in the text. The precision of computation was set to be $3\times10^{-8}.$}
    \label{tab::compa}
\end{table*}

\section{Conclusions}{\label{s::con}}

In this paper, we present \textbf{entcalc}~\cite{entcalc,entcalcpy}, a   new package for computing geometric entanglement. These package operates by evaluating lower and upper bounds of a geometric entanglement for a given multipartite state. We presented four different types of lower bounds and demonstrate that the gap between the upper and lower bounds can often be made very small. Consequently, this allows us to compute the geometric entanglement of a given state with high accuracy.

We demonstrated that our package can compute, with high accuracy, the amount of entanglement in the $3\otimes 3$ PPT entangled state. We then show that our package can be used for detecting and quantifying entanglement in bound entangled states that arise in certain spin-chain models. Such cases are particularly challenging because measures based on negativity fail to detect bound entanglement.
Furthermore, we showed that our package can accurately characterize thermal entanglement in states near a quantum phase transition. Using entcalc, we also demonstrated that multipartite entanglement can be generated between non-neighboring sites by tuning the magnetic field. Finally, we showed that entcalc can quantify entanglement in states obtained by applying noise to highly entangled states. The geometric entanglement was typically computed with precision $10^{-6}$ for a three qubit systems. 

We also compared the bounds used within this work and concluded that the most accurate results are obtained using the bound in \eqref{eq::bopurs}. For a balance between accuracy and computational efficiency, the bound in \eqref{eq::bopptk} with 
$k=2$ performs best. The bound in \eqref{eq::bopptf} is suitable for fast computations. Although it can sometimes produce high-precision estimates, it is generally less accurate.

We believe that our package can be a useful tool for studying entanglement and can help uncover new phenomena in spin chains, identify new bound entangled states, and support research on entanglement dynamics in noisy environments.

\section*{Acknowledgements}
We thank Pablo V. Parellada for a discussion. This work was supported by the National Science Centre Poland (Grant No. 2022/46/E/ST2/00115 and 2024/55/B/ST2/01590) and within the QuantERA II Programme (Grant No. 2021/03/Y/ST2/00178, acronym ExTRaQT) that has received funding from the European Union's Horizon 2020 research and innovation programme under Grant Agreement No. 101017733.
\bibliography{biblo}

\begin{thebibliography}{70}%
\makeatletter
\providecommand \@ifxundefined [1]{%
 \@ifx{#1\undefined}
}%
\providecommand \@ifnum [1]{%
 \ifnum #1\expandafter \@firstoftwo
 \else \expandafter \@secondoftwo
 \fi
}%
\providecommand \@ifx [1]{%
 \ifx #1\expandafter \@firstoftwo
 \else \expandafter \@secondoftwo
 \fi
}%
\providecommand \natexlab [1]{#1}%
\providecommand \enquote  [1]{``#1''}%
\providecommand \bibnamefont  [1]{#1}%
\providecommand \bibfnamefont [1]{#1}%
\providecommand \citenamefont [1]{#1}%
\providecommand \href@noop [0]{\@secondoftwo}%
\providecommand \href [0]{\begingroup \@sanitize@url \@href}%
\providecommand \@href[1]{\@@startlink{#1}\@@href}%
\providecommand \@@href[1]{\endgroup#1\@@endlink}%
\providecommand \@sanitize@url [0]{\catcode `\\12\catcode `\$12\catcode `\&12\catcode `\#12\catcode `\^12\catcode `\_12\catcode `\%12\relax}%
\providecommand \@@startlink[1]{}%
\providecommand \@@endlink[0]{}%
\providecommand \url  [0]{\begingroup\@sanitize@url \@url }%
\providecommand \@url [1]{\endgroup\@href {#1}{\urlprefix }}%
\providecommand \urlprefix  [0]{URL }%
\providecommand \Eprint [0]{\href }%
\providecommand \doibase [0]{https://doi.org/}%
\providecommand \selectlanguage [0]{\@gobble}%
\providecommand \bibinfo  [0]{\@secondoftwo}%
\providecommand \bibfield  [0]{\@secondoftwo}%
\providecommand \translation [1]{[#1]}%
\providecommand \BibitemOpen [0]{}%
\providecommand \bibitemStop [0]{}%
\providecommand \bibitemNoStop [0]{.\EOS\space}%
\providecommand \EOS [0]{\spacefactor3000\relax}%
\providecommand \BibitemShut  [1]{\csname bibitem#1\endcsname}%
\let\auto@bib@innerbib\@empty
\bibitem [{\citenamefont {Horodecki}\ \emph {et~al.}(2009)\citenamefont {Horodecki}, \citenamefont {Horodecki}, \citenamefont {Horodecki},\ and\ \citenamefont {Horodecki}}]{Horodeckientanglementreview}%
  \BibitemOpen
  \bibfield  {author} {\bibinfo {author} {\bibfnamefont {R.}~\bibnamefont {Horodecki}}, \bibinfo {author} {\bibfnamefont {P.}~\bibnamefont {Horodecki}}, \bibinfo {author} {\bibfnamefont {M.}~\bibnamefont {Horodecki}},\ and\ \bibinfo {author} {\bibfnamefont {K.}~\bibnamefont {Horodecki}},\ }\bibfield  {title} {\bibinfo {title} {Quantum entanglement},\ }\href {https://doi.org/10.1103/RevModPhys.81.865} {\bibfield  {journal} {\bibinfo  {journal} {Rev. Mod. Phys.}\ }\textbf {\bibinfo {volume} {81}},\ \bibinfo {pages} {865} (\bibinfo {year} {2009})}\BibitemShut {NoStop}%
\bibitem [{\citenamefont {Bennett}\ \emph {et~al.}(1993)\citenamefont {Bennett}, \citenamefont {Brassard}, \citenamefont {Cr\'epeau}, \citenamefont {Jozsa}, \citenamefont {Peres},\ and\ \citenamefont {Wootters}}]{Benettteleportofstate}%
  \BibitemOpen
  \bibfield  {author} {\bibinfo {author} {\bibfnamefont {C.~H.}\ \bibnamefont {Bennett}}, \bibinfo {author} {\bibfnamefont {G.}~\bibnamefont {Brassard}}, \bibinfo {author} {\bibfnamefont {C.}~\bibnamefont {Cr\'epeau}}, \bibinfo {author} {\bibfnamefont {R.}~\bibnamefont {Jozsa}}, \bibinfo {author} {\bibfnamefont {A.}~\bibnamefont {Peres}},\ and\ \bibinfo {author} {\bibfnamefont {W.~K.}\ \bibnamefont {Wootters}},\ }\bibfield  {title} {\bibinfo {title} {Teleporting an unknown quantum state via dual classical and {Einstein}-{Podolsky}-{Rosen} channels},\ }\href {https://doi.org/10.1103/PhysRevLett.70.1895} {\bibfield  {journal} {\bibinfo  {journal} {Phys. Rev. Lett.}\ }\textbf {\bibinfo {volume} {70}},\ \bibinfo {pages} {1895} (\bibinfo {year} {1993})}\BibitemShut {NoStop}%
\bibitem [{\citenamefont {Bennett}\ and\ \citenamefont {Wiesner}(1992)}]{Benettsuperdensecoding}%
  \BibitemOpen
  \bibfield  {author} {\bibinfo {author} {\bibfnamefont {C.~H.}\ \bibnamefont {Bennett}}\ and\ \bibinfo {author} {\bibfnamefont {S.~J.}\ \bibnamefont {Wiesner}},\ }\bibfield  {title} {\bibinfo {title} {Communication via one- and two-particle operators on {Einstein}-{Podolsky}-{Rosen} states},\ }\href {https://doi.org/10.1103/PhysRevLett.69.2881} {\bibfield  {journal} {\bibinfo  {journal} {Phys. Rev. Lett.}\ }\textbf {\bibinfo {volume} {69}},\ \bibinfo {pages} {2881} (\bibinfo {year} {1992})}\BibitemShut {NoStop}%
\bibitem [{\citenamefont {Bennett}\ and\ \citenamefont {Brassard}(1984)}]{BennettQKD}%
  \BibitemOpen
  \bibfield  {author} {\bibinfo {author} {\bibfnamefont {C.~H.}\ \bibnamefont {Bennett}}\ and\ \bibinfo {author} {\bibfnamefont {G.}~\bibnamefont {Brassard}},\ }\bibfield  {title} {\bibinfo {title} {Quantum cryptography: Public key distribution and coin tossing},\ }in\ \href {https://arxiv.org/abs/1402.4854} {\emph {\bibinfo {booktitle} {Proceedings of IEEE International Conference on Computers, Systems and Signal Processing}}}\ (\bibinfo {address} {Bangalore, India},\ \bibinfo {year} {1984})\ pp.\ \bibinfo {pages} {175--179}\BibitemShut {NoStop}%
\bibitem [{\citenamefont {Nielsen}\ and\ \citenamefont {Chuang}(2010)}]{Nielsencambridgeup2010}%
  \BibitemOpen
  \bibfield  {author} {\bibinfo {author} {\bibfnamefont {M.~A.}\ \bibnamefont {Nielsen}}\ and\ \bibinfo {author} {\bibfnamefont {I.~L.}\ \bibnamefont {Chuang}},\ }\href {https://doi.org/https://doi.org/10.1017/CBO9780511976667} {\emph {\bibinfo {title} {Quantum Computation and Quantum Information: 10th Anniversary Edition}}}\ (\bibinfo  {publisher} {Cambridge University Press},\ \bibinfo {year} {2010})\BibitemShut {NoStop}%
\bibitem [{\citenamefont {Vidal}\ and\ \citenamefont {Werner}(2002)}]{Vidalneg}%
  \BibitemOpen
  \bibfield  {author} {\bibinfo {author} {\bibfnamefont {G.}~\bibnamefont {Vidal}}\ and\ \bibinfo {author} {\bibfnamefont {R.~F.}\ \bibnamefont {Werner}},\ }\bibfield  {title} {\bibinfo {title} {Computable measure of entanglement},\ }\href {https://doi.org/10.1103/PhysRevA.65.032314} {\bibfield  {journal} {\bibinfo  {journal} {Phys. Rev. A}\ }\textbf {\bibinfo {volume} {65}},\ \bibinfo {pages} {032314} (\bibinfo {year} {2002})}\BibitemShut {NoStop}%
\bibitem [{\citenamefont {Wei}\ and\ \citenamefont {Goldbart}(2003)}]{Weipra68_2003}%
  \BibitemOpen
  \bibfield  {author} {\bibinfo {author} {\bibfnamefont {T.-C.}\ \bibnamefont {Wei}}\ and\ \bibinfo {author} {\bibfnamefont {P.~M.}\ \bibnamefont {Goldbart}},\ }\bibfield  {title} {\bibinfo {title} {Geometric measure of entanglement and applications to bipartite and multipartite quantum states},\ }\href {https://doi.org/10.1103/PhysRevA.68.042307} {\bibfield  {journal} {\bibinfo  {journal} {Phys. Rev. A}\ }\textbf {\bibinfo {volume} {68}},\ \bibinfo {pages} {042307} (\bibinfo {year} {2003})}\BibitemShut {NoStop}%
\bibitem [{\citenamefont {Bennett}\ \emph {et~al.}(1996)\citenamefont {Bennett}, \citenamefont {DiVincenzo}, \citenamefont {Smolin},\ and\ \citenamefont {Wootters}}]{BenettLOCC}%
  \BibitemOpen
  \bibfield  {author} {\bibinfo {author} {\bibfnamefont {C.~H.}\ \bibnamefont {Bennett}}, \bibinfo {author} {\bibfnamefont {D.~P.}\ \bibnamefont {DiVincenzo}}, \bibinfo {author} {\bibfnamefont {J.~A.}\ \bibnamefont {Smolin}},\ and\ \bibinfo {author} {\bibfnamefont {W.~K.}\ \bibnamefont {Wootters}},\ }\bibfield  {title} {\bibinfo {title} {{Mixed-state entanglement and quantum error correction}},\ }\href {https://doi.org/10.1103/PhysRevA.54.3824} {\bibfield  {journal} {\bibinfo  {journal} {Phys. Rev. A}\ }\textbf {\bibinfo {volume} {54}},\ \bibinfo {pages} {3824} (\bibinfo {year} {1996})}\BibitemShut {NoStop}%
\bibitem [{\citenamefont {Gurvits}(2003)}]{gurvitssep}%
  \BibitemOpen
  \bibfield  {author} {\bibinfo {author} {\bibfnamefont {L.}~\bibnamefont {Gurvits}},\ }\href {https://arxiv.org/abs/quant-ph/0303055} {\bibinfo {title} {{Classical deterministic complexity of Edmonds' problem and Quantum Entanglement}}} (\bibinfo {year} {2003}),\ \Eprint {https://arxiv.org/abs/quant-ph/0303055} {arXiv:quant-ph/0303055 [quant-ph]} \BibitemShut {NoStop}%
\bibitem [{\citenamefont {Gharibian}(2010)}]{Gharibianse}%
  \BibitemOpen
  \bibfield  {author} {\bibinfo {author} {\bibfnamefont {S.}~\bibnamefont {Gharibian}},\ }\bibfield  {title} {\bibinfo {title} {Strong {NP}-hardness of the quantum separability problem},\ }\href {https://doi.org/10.26421/QIC10.3-4-11} {\bibfield  {journal} {\bibinfo  {journal} {Quantum Inf. Comput.}\ }\textbf {\bibinfo {volume} {10}},\ \bibinfo {pages} {343} (\bibinfo {year} {2010})}\BibitemShut {NoStop}%
\bibitem [{\citenamefont {Horodecki}\ \emph {et~al.}(1996)\citenamefont {Horodecki}, \citenamefont {Horodecki},\ and\ \citenamefont {Horodecki}}]{Horodeckipptcriterion}%
  \BibitemOpen
  \bibfield  {author} {\bibinfo {author} {\bibfnamefont {M.}~\bibnamefont {Horodecki}}, \bibinfo {author} {\bibfnamefont {P.}~\bibnamefont {Horodecki}},\ and\ \bibinfo {author} {\bibfnamefont {R.}~\bibnamefont {Horodecki}},\ }\bibfield  {title} {\bibinfo {title} {Separability of mixed states: necessary and sufficient conditions},\ }\href {https://doi.org/https://doi.org/10.1016/S0375-9601(96)00706-2} {\bibfield  {journal} {\bibinfo  {journal} {Physics Letters A}\ }\textbf {\bibinfo {volume} {223}},\ \bibinfo {pages} {1} (\bibinfo {year} {1996})}\BibitemShut {NoStop}%
\bibitem [{\citenamefont {Peres}(1996)}]{Perespptcriterion}%
  \BibitemOpen
  \bibfield  {author} {\bibinfo {author} {\bibfnamefont {A.}~\bibnamefont {Peres}},\ }\bibfield  {title} {\bibinfo {title} {Separability criterion for density matrices},\ }\href {https://doi.org/10.1103/PhysRevLett.77.1413} {\bibfield  {journal} {\bibinfo  {journal} {Phys. Rev. Lett.}\ }\textbf {\bibinfo {volume} {77}},\ \bibinfo {pages} {1413} (\bibinfo {year} {1996})}\BibitemShut {NoStop}%
\bibitem [{\citenamefont {Terhal}(2002)}]{Terhaldetecting}%
  \BibitemOpen
  \bibfield  {author} {\bibinfo {author} {\bibfnamefont {B.~M.}\ \bibnamefont {Terhal}},\ }\bibfield  {title} {\bibinfo {title} {Detecting quantum entanglement},\ }\href {https://doi.org/https://doi.org/10.1016/S0304-3975(02)00139-1} {\bibfield  {journal} {\bibinfo  {journal} {Theoretical computer science}\ }\textbf {\bibinfo {volume} {287}},\ \bibinfo {pages} {313} (\bibinfo {year} {2002})}\BibitemShut {NoStop}%
\bibitem [{\citenamefont {Lee}\ \emph {et~al.}(2002)\citenamefont {Lee}, \citenamefont {Min},\ and\ \citenamefont {Oh}}]{Leemul}%
  \BibitemOpen
  \bibfield  {author} {\bibinfo {author} {\bibfnamefont {J.}~\bibnamefont {Lee}}, \bibinfo {author} {\bibfnamefont {H.}~\bibnamefont {Min}},\ and\ \bibinfo {author} {\bibfnamefont {S.~D.}\ \bibnamefont {Oh}},\ }\bibfield  {title} {\bibinfo {title} {Multipartite entanglement for entanglement teleportation},\ }\href {https://doi.org/10.1103/PhysRevA.66.052318} {\bibfield  {journal} {\bibinfo  {journal} {Phys. Rev. A}\ }\textbf {\bibinfo {volume} {66}},\ \bibinfo {pages} {052318} (\bibinfo {year} {2002})}\BibitemShut {NoStop}%
\bibitem [{\citenamefont {Epping}\ \emph {et~al.}(2017)\citenamefont {Epping}, \citenamefont {Kampermann}, \citenamefont {macchiavello},\ and\ \citenamefont {Bruß}}]{Eppingmul}%
  \BibitemOpen
  \bibfield  {author} {\bibinfo {author} {\bibfnamefont {M.}~\bibnamefont {Epping}}, \bibinfo {author} {\bibfnamefont {H.}~\bibnamefont {Kampermann}}, \bibinfo {author} {\bibfnamefont {C.}~\bibnamefont {macchiavello}},\ and\ \bibinfo {author} {\bibfnamefont {D.}~\bibnamefont {Bruß}},\ }\bibfield  {title} {\bibinfo {title} {Multi-partite entanglement can speed up quantum key distribution in networks},\ }\href {https://doi.org/10.1088/1367-2630/aa8487} {\bibfield  {journal} {\bibinfo  {journal} {New Journal of Physics}\ }\textbf {\bibinfo {volume} {19}},\ \bibinfo {pages} {093012} (\bibinfo {year} {2017})}\BibitemShut {NoStop}%
\bibitem [{\citenamefont {Svetlichny}(1987)}]{Svetlichnyge}%
  \BibitemOpen
  \bibfield  {author} {\bibinfo {author} {\bibfnamefont {G.}~\bibnamefont {Svetlichny}},\ }\bibfield  {title} {\bibinfo {title} {{Distinguishing three-body from two-body nonseparability by a Bell-type inequality}},\ }\href {https://doi.org/10.1103/PhysRevD.35.3066} {\bibfield  {journal} {\bibinfo  {journal} {Phys. Rev. D}\ }\textbf {\bibinfo {volume} {35}},\ \bibinfo {pages} {3066} (\bibinfo {year} {1987})}\BibitemShut {NoStop}%
\bibitem [{\citenamefont {Fitter}\ \emph {et~al.}(2025)\citenamefont {Fitter}, \citenamefont {Lancien},\ and\ \citenamefont {Nechita}}]{Fitterestimatingentanglementrandommultipartite}%
  \BibitemOpen
  \bibfield  {author} {\bibinfo {author} {\bibfnamefont {K.~P.}\ \bibnamefont {Fitter}}, \bibinfo {author} {\bibfnamefont {C.}~\bibnamefont {Lancien}},\ and\ \bibinfo {author} {\bibfnamefont {I.}~\bibnamefont {Nechita}},\ }\href {https://arxiv.org/abs/2209.11754} {\bibinfo {title} {Estimating the entanglement of random multipartite quantum states}} (\bibinfo {year} {2025}),\ \Eprint {https://arxiv.org/abs/2209.11754} {arXiv:2209.11754 [quant-ph]} \BibitemShut {NoStop}%
\bibitem [{\citenamefont {Wang}\ \emph {et~al.}(2025)\citenamefont {Wang}, \citenamefont {Ma}, \citenamefont {Chen}, \citenamefont {Zhang},\ and\ \citenamefont {Fei}}]{Wanggme}%
  \BibitemOpen
  \bibfield  {author} {\bibinfo {author} {\bibfnamefont {Z.}~\bibnamefont {Wang}}, \bibinfo {author} {\bibfnamefont {Z.}~\bibnamefont {Ma}}, \bibinfo {author} {\bibfnamefont {L.}~\bibnamefont {Chen}}, \bibinfo {author} {\bibfnamefont {C.}~\bibnamefont {Zhang}},\ and\ \bibinfo {author} {\bibfnamefont {S.-M.}\ \bibnamefont {Fei}},\ }\bibfield  {title} {\bibinfo {title} {{Unified Construction of Genuine Multipartite Entanglement Measures Based on Geometric Mean and its Applications}},\ }\bibfield  {journal} {\bibinfo  {journal} {Quantum Information Processing}\ }\href {https://doi.org/https://doi.org/10.1007/s11128-025-04862-y} {https://doi.org/10.1007/s11128-025-04862-y} (\bibinfo {year} {2025}),\ \Eprint {https://arxiv.org/abs/arXiv:2508.10301} {arXiv:2508.10301} \BibitemShut {NoStop}%
\bibitem [{\citenamefont {Shi-Rong}\ \emph {et~al.}(2010)\citenamefont {Shi-Rong}, \citenamefont {Yun-Jie},\ and\ \citenamefont {Zhong-Xiao}}]{Chenqft}%
  \BibitemOpen
  \bibfield  {author} {\bibinfo {author} {\bibfnamefont {C.}~\bibnamefont {Shi-Rong}}, \bibinfo {author} {\bibfnamefont {X.}~\bibnamefont {Yun-Jie}},\ and\ \bibinfo {author} {\bibfnamefont {M.}~\bibnamefont {Zhong-Xiao}},\ }\bibfield  {title} {\bibinfo {title} {{Quantum phase transition and entanglement in Heisenberg XX spin chain with impurity}},\ }\href {https://doi.org/10.1088/1674-1056/19/5/050304} {\bibfield  {journal} {\bibinfo  {journal} {Chinese Physics B}\ }\textbf {\bibinfo {volume} {19}},\ \bibinfo {pages} {050304} (\bibinfo {year} {2010})}\BibitemShut {NoStop}%
\bibitem [{\citenamefont {Osborne}\ and\ \citenamefont {Nielsen}(2002)}]{Osborneqft}%
  \BibitemOpen
  \bibfield  {author} {\bibinfo {author} {\bibfnamefont {T.~J.}\ \bibnamefont {Osborne}}\ and\ \bibinfo {author} {\bibfnamefont {M.~A.}\ \bibnamefont {Nielsen}},\ }\bibfield  {title} {\bibinfo {title} {Entanglement in a simple quantum phase transition},\ }\href {https://doi.org/10.1103/PhysRevA.66.032110} {\bibfield  {journal} {\bibinfo  {journal} {Phys. Rev. A}\ }\textbf {\bibinfo {volume} {66}},\ \bibinfo {pages} {032110} (\bibinfo {year} {2002})}\BibitemShut {NoStop}%
\bibitem [{\citenamefont {Greenberger}\ \emph {et~al.}(1989)\citenamefont {Greenberger}, \citenamefont {Horne},\ and\ \citenamefont {Zeilinger}}]{Greenbergerghzsta}%
  \BibitemOpen
  \bibfield  {author} {\bibinfo {author} {\bibfnamefont {D.~M.}\ \bibnamefont {Greenberger}}, \bibinfo {author} {\bibfnamefont {M.~A.}\ \bibnamefont {Horne}},\ and\ \bibinfo {author} {\bibfnamefont {A.}~\bibnamefont {Zeilinger}},\ }\bibinfo {title} {Going beyond bell's theorem},\ in\ \href {https://doi.org/10.1007/978-94-017-0849-4_10} {\emph {\bibinfo {booktitle} {Bell's Theorem, Quantum Theory and Conceptions of the Universe}}},\ \bibinfo {editor} {edited by\ \bibinfo {editor} {\bibfnamefont {M.}~\bibnamefont {Kafatos}}}\ (\bibinfo  {publisher} {Springer Netherlands},\ \bibinfo {address} {Dordrecht},\ \bibinfo {year} {1989})\ pp.\ \bibinfo {pages} {69--72}\BibitemShut {NoStop}%
\bibitem [{\citenamefont {Gao}\ and\ \citenamefont {Hong}(2010)}]{Gaomulwitn}%
  \BibitemOpen
  \bibfield  {author} {\bibinfo {author} {\bibfnamefont {T.}~\bibnamefont {Gao}}\ and\ \bibinfo {author} {\bibfnamefont {Y.}~\bibnamefont {Hong}},\ }\bibfield  {title} {\bibinfo {title} {Detection of genuinely entangled and nonseparable $n$-partite quantum states},\ }\href {https://doi.org/10.1103/PhysRevA.82.062113} {\bibfield  {journal} {\bibinfo  {journal} {Phys. Rev. A}\ }\textbf {\bibinfo {volume} {82}},\ \bibinfo {pages} {062113} (\bibinfo {year} {2010})}\BibitemShut {NoStop}%
\bibitem [{\citenamefont {Bourennane}\ \emph {et~al.}(2004)\citenamefont {Bourennane}, \citenamefont {Eibl}, \citenamefont {Kurtsiefer}, \citenamefont {Gaertner}, \citenamefont {Weinfurter}, \citenamefont {G\"uhne}, \citenamefont {Hyllus}, \citenamefont {Bru\ss{}}, \citenamefont {Lewenstein},\ and\ \citenamefont {Sanpera}}]{Bournamemulwit}%
  \BibitemOpen
  \bibfield  {author} {\bibinfo {author} {\bibfnamefont {M.}~\bibnamefont {Bourennane}}, \bibinfo {author} {\bibfnamefont {M.}~\bibnamefont {Eibl}}, \bibinfo {author} {\bibfnamefont {C.}~\bibnamefont {Kurtsiefer}}, \bibinfo {author} {\bibfnamefont {S.}~\bibnamefont {Gaertner}}, \bibinfo {author} {\bibfnamefont {H.}~\bibnamefont {Weinfurter}}, \bibinfo {author} {\bibfnamefont {O.}~\bibnamefont {G\"uhne}}, \bibinfo {author} {\bibfnamefont {P.}~\bibnamefont {Hyllus}}, \bibinfo {author} {\bibfnamefont {D.}~\bibnamefont {Bru\ss{}}}, \bibinfo {author} {\bibfnamefont {M.}~\bibnamefont {Lewenstein}},\ and\ \bibinfo {author} {\bibfnamefont {A.}~\bibnamefont {Sanpera}},\ }\bibfield  {title} {\bibinfo {title} {Experimental detection of multipartite entanglement using witness operators},\ }\href {https://doi.org/10.1103/PhysRevLett.92.087902} {\bibfield  {journal} {\bibinfo  {journal} {Phys. Rev. Lett.}\ }\textbf {\bibinfo {volume} {92}},\ \bibinfo {pages} {087902} (\bibinfo {year} {2004})}\BibitemShut {NoStop}%
\bibitem [{\citenamefont {Ma}\ \emph {et~al.}(2011)\citenamefont {Ma}, \citenamefont {Chen}, \citenamefont {Chen}, \citenamefont {Spengler}, \citenamefont {Gabriel},\ and\ \citenamefont {Huber}}]{Magme}%
  \BibitemOpen
  \bibfield  {author} {\bibinfo {author} {\bibfnamefont {Z.-H.}\ \bibnamefont {Ma}}, \bibinfo {author} {\bibfnamefont {Z.-H.}\ \bibnamefont {Chen}}, \bibinfo {author} {\bibfnamefont {J.-L.}\ \bibnamefont {Chen}}, \bibinfo {author} {\bibfnamefont {C.}~\bibnamefont {Spengler}}, \bibinfo {author} {\bibfnamefont {A.}~\bibnamefont {Gabriel}},\ and\ \bibinfo {author} {\bibfnamefont {M.}~\bibnamefont {Huber}},\ }\bibfield  {title} {\bibinfo {title} {Measure of genuine multipartite entanglement with computable lower bounds},\ }\href {https://doi.org/10.1103/PhysRevA.83.062325} {\bibfield  {journal} {\bibinfo  {journal} {Phys. Rev. A}\ }\textbf {\bibinfo {volume} {83}},\ \bibinfo {pages} {062325} (\bibinfo {year} {2011})}\BibitemShut {NoStop}%
\bibitem [{\citenamefont {Jungnitsch}\ \emph {et~al.}(2011)\citenamefont {Jungnitsch}, \citenamefont {Moroder},\ and\ \citenamefont {G\"uhne}}]{Jungnitschlower}%
  \BibitemOpen
  \bibfield  {author} {\bibinfo {author} {\bibfnamefont {B.}~\bibnamefont {Jungnitsch}}, \bibinfo {author} {\bibfnamefont {T.}~\bibnamefont {Moroder}},\ and\ \bibinfo {author} {\bibfnamefont {O.}~\bibnamefont {G\"uhne}},\ }\bibfield  {title} {\bibinfo {title} {Taming multiparticle entanglement},\ }\href {https://doi.org/10.1103/PhysRevLett.106.190502} {\bibfield  {journal} {\bibinfo  {journal} {Phys. Rev. Lett.}\ }\textbf {\bibinfo {volume} {106}},\ \bibinfo {pages} {190502} (\bibinfo {year} {2011})}\BibitemShut {NoStop}%
\bibitem [{\citenamefont {Li}\ \emph {et~al.}(2024)\citenamefont {Li}, \citenamefont {Dong}, \citenamefont {Zhang}, \citenamefont {Zhu}, \citenamefont {Shen}, \citenamefont {Li},\ and\ \citenamefont {Fei}}]{Ligme}%
  \BibitemOpen
  \bibfield  {author} {\bibinfo {author} {\bibfnamefont {M.}~\bibnamefont {Li}}, \bibinfo {author} {\bibfnamefont {Y.}~\bibnamefont {Dong}}, \bibinfo {author} {\bibfnamefont {R.}~\bibnamefont {Zhang}}, \bibinfo {author} {\bibfnamefont {X.}~\bibnamefont {Zhu}}, \bibinfo {author} {\bibfnamefont {S.}~\bibnamefont {Shen}}, \bibinfo {author} {\bibfnamefont {L.}~\bibnamefont {Li}},\ and\ \bibinfo {author} {\bibfnamefont {S.-M.}\ \bibnamefont {Fei}},\ }\bibfield  {title} {\bibinfo {title} {A note on the lower bounds of genuine multipartite entanglement concurrence},\ }\href {https://doi.org/10.1007/s11128-024-04607-3} {\bibfield  {journal} {\bibinfo  {journal} {Quantum Information Processing}\ }\textbf {\bibinfo {volume} {23}},\ \bibinfo {pages} {397} (\bibinfo {year} {2024})}\BibitemShut {NoStop}%
\bibitem [{\citenamefont {Dai}\ \emph {et~al.}(2020)\citenamefont {Dai}, \citenamefont {Dong}, \citenamefont {Xu}, \citenamefont {You}, \citenamefont {Zhang},\ and\ \citenamefont {G\"uhne}}]{Daigme}%
  \BibitemOpen
  \bibfield  {author} {\bibinfo {author} {\bibfnamefont {Y.}~\bibnamefont {Dai}}, \bibinfo {author} {\bibfnamefont {Y.}~\bibnamefont {Dong}}, \bibinfo {author} {\bibfnamefont {Z.}~\bibnamefont {Xu}}, \bibinfo {author} {\bibfnamefont {W.}~\bibnamefont {You}}, \bibinfo {author} {\bibfnamefont {C.}~\bibnamefont {Zhang}},\ and\ \bibinfo {author} {\bibfnamefont {O.}~\bibnamefont {G\"uhne}},\ }\bibfield  {title} {\bibinfo {title} {Experimentally accessible lower bounds for genuine multipartite entanglement and coherence measures},\ }\href {https://doi.org/10.1103/PhysRevApplied.13.054022} {\bibfield  {journal} {\bibinfo  {journal} {Phys. Rev. Appl.}\ }\textbf {\bibinfo {volume} {13}},\ \bibinfo {pages} {054022} (\bibinfo {year} {2020})}\BibitemShut {NoStop}%
\bibitem [{\citenamefont {Streltsov}\ \emph {et~al.}(2011)\citenamefont {Streltsov}, \citenamefont {Kampermann},\ and\ \citenamefont {Bru\ss{}}}]{Streltsovupperbound}%
  \BibitemOpen
  \bibfield  {author} {\bibinfo {author} {\bibfnamefont {A.}~\bibnamefont {Streltsov}}, \bibinfo {author} {\bibfnamefont {H.}~\bibnamefont {Kampermann}},\ and\ \bibinfo {author} {\bibfnamefont {D.}~\bibnamefont {Bru\ss{}}},\ }\bibfield  {title} {\bibinfo {title} {Simple algorithm for computing the geometric measure of entanglement},\ }\href {https://doi.org/10.1103/PhysRevA.84.022323} {\bibfield  {journal} {\bibinfo  {journal} {Phys. Rev. A}\ }\textbf {\bibinfo {volume} {84}},\ \bibinfo {pages} {022323} (\bibinfo {year} {2011})}\BibitemShut {NoStop}%
\bibitem [{\citenamefont {Ganardi}\ \emph {et~al.}(2025)\citenamefont {Ganardi}, \citenamefont {Masajada}, \citenamefont {Naseri},\ and\ \citenamefont {Streltsov}}]{Ganardilocalpurity}%
  \BibitemOpen
  \bibfield  {author} {\bibinfo {author} {\bibfnamefont {R.}~\bibnamefont {Ganardi}}, \bibinfo {author} {\bibfnamefont {P.}~\bibnamefont {Masajada}}, \bibinfo {author} {\bibfnamefont {M.}~\bibnamefont {Naseri}},\ and\ \bibinfo {author} {\bibfnamefont {A.}~\bibnamefont {Streltsov}},\ }\bibfield  {title} {\bibinfo {title} {Local {P}urity {D}istillation in {Q}uantum {S}ystems: {E}xploring the {C}omplementarity {B}etween {P}urity and {E}ntanglement},\ }\href {https://doi.org/10.22331/q-2025-03-20-1666} {\bibfield  {journal} {\bibinfo  {journal} {{Quantum}}\ }\textbf {\bibinfo {volume} {9}},\ \bibinfo {pages} {1666} (\bibinfo {year} {2025})}\BibitemShut {NoStop}%
\bibitem [{\citenamefont {Streltsov}\ \emph {et~al.}(2010)\citenamefont {Streltsov}, \citenamefont {Kampermann},\ and\ \citenamefont {Bru{\ss}}}]{Streltsov_2010}%
  \BibitemOpen
  \bibfield  {author} {\bibinfo {author} {\bibfnamefont {A.}~\bibnamefont {Streltsov}}, \bibinfo {author} {\bibfnamefont {H.}~\bibnamefont {Kampermann}},\ and\ \bibinfo {author} {\bibfnamefont {D.}~\bibnamefont {Bru{\ss}}},\ }\bibfield  {title} {\bibinfo {title} {Linking a distance measure of entanglement to its convex roof},\ }\href {https://doi.org/10.1088/1367-2630/12/12/123004} {\bibfield  {journal} {\bibinfo  {journal} {New Journal of Physics}\ }\textbf {\bibinfo {volume} {12}},\ \bibinfo {pages} {123004} (\bibinfo {year} {2010})}\BibitemShut {NoStop}%
\bibitem [{\citenamefont {Doherty}\ \emph {et~al.}(2004)\citenamefont {Doherty}, \citenamefont {Parrilo},\ and\ \citenamefont {Spedalieri}}]{Dohertyksymmetric}%
  \BibitemOpen
  \bibfield  {author} {\bibinfo {author} {\bibfnamefont {A.~C.}\ \bibnamefont {Doherty}}, \bibinfo {author} {\bibfnamefont {P.~A.}\ \bibnamefont {Parrilo}},\ and\ \bibinfo {author} {\bibfnamefont {F.~M.}\ \bibnamefont {Spedalieri}},\ }\bibfield  {title} {\bibinfo {title} {Complete family of separability criteria},\ }\href {https://doi.org/10.1103/PhysRevA.69.022308} {\bibfield  {journal} {\bibinfo  {journal} {Phys. Rev. A}\ }\textbf {\bibinfo {volume} {69}},\ \bibinfo {pages} {022308} (\bibinfo {year} {2004})}\BibitemShut {NoStop}%
\bibitem [{ent(2025{\natexlab{a}})}]{entcalc}%
  \BibitemOpen
  \href@noop {} {\bibinfo {title} {entcalc}},\ \bibinfo {howpublished} {Package can be viewed at this link: \url{https://github.com/piotrero00/entcalc}} (\bibinfo {year} {2025}{\natexlab{a}})\BibitemShut {NoStop}%
\bibitem [{ent(2025{\natexlab{b}})}]{entcalcpy}%
  \BibitemOpen
  \href@noop {} {\bibinfo {title} {entcalcpy}},\ \bibinfo {howpublished} {Package can be viewed at this link: \url{https://github.com/piotrero00/entcalcpy}} (\bibinfo {year} {2025}{\natexlab{b}})\BibitemShut {NoStop}%
\bibitem [{\citenamefont {{CVX Research, Inc.}}(2012)}]{cvx}%
  \BibitemOpen
  \bibfield  {author} {\bibinfo {author} {\bibnamefont {{CVX Research, Inc.}}},\ }\href@noop {} {\bibinfo {title} {{CVX}: Matlab software for disciplined convex programming, version 2.0}},\ \bibinfo {howpublished} {\url{https://cvxr.com/cvx}} (\bibinfo {year} {2012})\BibitemShut {NoStop}%
\bibitem [{\citenamefont {Grant}\ and\ \citenamefont {Boyd}(2008)}]{Grantcvx}%
  \BibitemOpen
  \bibfield  {author} {\bibinfo {author} {\bibfnamefont {M.}~\bibnamefont {Grant}}\ and\ \bibinfo {author} {\bibfnamefont {S.}~\bibnamefont {Boyd}},\ }\bibfield  {title} {\bibinfo {title} {Graph implementations for nonsmooth convex programs},\ }in\ \href@noop {} {\emph {\bibinfo {booktitle} {Recent Advances in Learning and Control}}},\ \bibinfo {series and number} {Lecture Notes in Control and Information Sciences},\ \bibinfo {editor} {edited by\ \bibinfo {editor} {\bibfnamefont {V.}~\bibnamefont {Blondel}}, \bibinfo {editor} {\bibfnamefont {S.}~\bibnamefont {Boyd}},\ and\ \bibinfo {editor} {\bibfnamefont {H.}~\bibnamefont {Kimura}}}\ (\bibinfo  {publisher} {Springer-Verlag Limited},\ \bibinfo {year} {2008})\ pp.\ \bibinfo {pages} {95--110},\ \bibinfo {note} {\url{http://stanford.edu/~boyd/graph_dcp.html}}\BibitemShut {NoStop}%
\bibitem [{\citenamefont {Johnston}(2016)}]{Johnstonqetlabpackagematlab}%
  \BibitemOpen
  \bibfield  {author} {\bibinfo {author} {\bibfnamefont {N.}~\bibnamefont {Johnston}},\ }\href {https://doi.org/10.5281/zenodo.44637} {\bibinfo {title} {{QETLAB}: A {MATLAB} toolbox for quantum entanglement, version 1.0}},\ \bibinfo {howpublished} {\url{https://qetlab.com}} (\bibinfo {year} {2016})\BibitemShut {NoStop}%
\bibitem [{\citenamefont {Harris}\ \emph {et~al.}(2020)\citenamefont {Harris}, \citenamefont {Millman}, \citenamefont {van~der Walt}, \citenamefont {Gommers}, \citenamefont {Virtanen}, \citenamefont {Cournapeau}, \citenamefont {Wieser}, \citenamefont {Taylor}, \citenamefont {Berg}, \citenamefont {Smith}, \citenamefont {Kern}, \citenamefont {Picus}, \citenamefont {Hoyer}, \citenamefont {van Kerkwijk}, \citenamefont {Brett}, \citenamefont {Haldane}, \citenamefont {del R{\'{i}}o}, \citenamefont {Wiebe}, \citenamefont {Peterson}, \citenamefont {G{\'{e}}rard-Marchant}, \citenamefont {Sheppard}, \citenamefont {Reddy}, \citenamefont {Weckesser}, \citenamefont {Abbasi}, \citenamefont {Gohlke},\ and\ \citenamefont {Oliphant}}]{Harrisnumpy}%
  \BibitemOpen
  \bibfield  {author} {\bibinfo {author} {\bibfnamefont {C.~R.}\ \bibnamefont {Harris}}, \bibinfo {author} {\bibfnamefont {K.~J.}\ \bibnamefont {Millman}}, \bibinfo {author} {\bibfnamefont {S.~J.}\ \bibnamefont {van~der Walt}}, \bibinfo {author} {\bibfnamefont {R.}~\bibnamefont {Gommers}}, \bibinfo {author} {\bibfnamefont {P.}~\bibnamefont {Virtanen}}, \bibinfo {author} {\bibfnamefont {D.}~\bibnamefont {Cournapeau}}, \bibinfo {author} {\bibfnamefont {E.}~\bibnamefont {Wieser}}, \bibinfo {author} {\bibfnamefont {J.}~\bibnamefont {Taylor}}, \bibinfo {author} {\bibfnamefont {S.}~\bibnamefont {Berg}}, \bibinfo {author} {\bibfnamefont {N.~J.}\ \bibnamefont {Smith}}, \bibinfo {author} {\bibfnamefont {R.}~\bibnamefont {Kern}}, \bibinfo {author} {\bibfnamefont {M.}~\bibnamefont {Picus}}, \bibinfo {author} {\bibfnamefont {S.}~\bibnamefont {Hoyer}}, \bibinfo {author} {\bibfnamefont {M.~H.}\ \bibnamefont {van Kerkwijk}}, \bibinfo {author} {\bibfnamefont {M.}~\bibnamefont {Brett}}, \bibinfo {author} {\bibfnamefont
  {A.}~\bibnamefont {Haldane}}, \bibinfo {author} {\bibfnamefont {J.~F.}\ \bibnamefont {del R{\'{i}}o}}, \bibinfo {author} {\bibfnamefont {M.}~\bibnamefont {Wiebe}}, \bibinfo {author} {\bibfnamefont {P.}~\bibnamefont {Peterson}}, \bibinfo {author} {\bibfnamefont {P.}~\bibnamefont {G{\'{e}}rard-Marchant}}, \bibinfo {author} {\bibfnamefont {K.}~\bibnamefont {Sheppard}}, \bibinfo {author} {\bibfnamefont {T.}~\bibnamefont {Reddy}}, \bibinfo {author} {\bibfnamefont {W.}~\bibnamefont {Weckesser}}, \bibinfo {author} {\bibfnamefont {H.}~\bibnamefont {Abbasi}}, \bibinfo {author} {\bibfnamefont {C.}~\bibnamefont {Gohlke}},\ and\ \bibinfo {author} {\bibfnamefont {T.~E.}\ \bibnamefont {Oliphant}},\ }\bibfield  {title} {\bibinfo {title} {Array programming with {NumPy}},\ }\href {https://doi.org/10.1038/s41586-020-2649-2} {\bibfield  {journal} {\bibinfo  {journal} {Nature}\ }\textbf {\bibinfo {volume} {585}},\ \bibinfo {pages} {357} (\bibinfo {year} {2020})}\BibitemShut {NoStop}%
\bibitem [{\citenamefont {Johansson}\ \emph {et~al.}(2013)\citenamefont {Johansson}, \citenamefont {Nation},\ and\ \citenamefont {Nori}}]{JOHANSSONqutippython}%
  \BibitemOpen
  \bibfield  {author} {\bibinfo {author} {\bibfnamefont {J.}~\bibnamefont {Johansson}}, \bibinfo {author} {\bibfnamefont {P.}~\bibnamefont {Nation}},\ and\ \bibinfo {author} {\bibfnamefont {F.}~\bibnamefont {Nori}},\ }\bibfield  {title} {\bibinfo {title} {{QuTiP 2: A Python framework for the dynamics of open quantum systems}},\ }\href {https://doi.org/https://doi.org/10.1016/j.cpc.2012.11.019} {\bibfield  {journal} {\bibinfo  {journal} {Computer Physics Communications}\ }\textbf {\bibinfo {volume} {184}},\ \bibinfo {pages} {1234} (\bibinfo {year} {2013})}\BibitemShut {NoStop}%
\bibitem [{\citenamefont {Diamond}\ and\ \citenamefont {Boyd}(2016)}]{Diamond2016cvxpy}%
  \BibitemOpen
  \bibfield  {author} {\bibinfo {author} {\bibfnamefont {S.}~\bibnamefont {Diamond}}\ and\ \bibinfo {author} {\bibfnamefont {S.}~\bibnamefont {Boyd}},\ }\bibfield  {title} {\bibinfo {title} {{CVXPY}: {A} {P}ython-embedded modeling language for convex optimization},\ }\href {https://doi.org/10.48550/arXiv.1603.00943} {\bibfield  {journal} {\bibinfo  {journal} {Journal of Machine Learning Research}\ }\textbf {\bibinfo {volume} {17}},\ \bibinfo {pages} {1} (\bibinfo {year} {2016})}\BibitemShut {NoStop}%
\bibitem [{\citenamefont {Virtanen}\ \emph {et~al.}(2020)\citenamefont {Virtanen}, \citenamefont {Gommers}, \citenamefont {Oliphant}, \citenamefont {Haberland}, \citenamefont {Reddy}, \citenamefont {Cournapeau}, \citenamefont {Burovski}, \citenamefont {Peterson}, \citenamefont {Weckesser}, \citenamefont {Bright}, \citenamefont {{van der Walt}}, \citenamefont {Brett}, \citenamefont {Wilson}, \citenamefont {Millman}, \citenamefont {Mayorov}, \citenamefont {Nelson}, \citenamefont {Jones}, \citenamefont {Kern}, \citenamefont {Larson}, \citenamefont {Carey}, \citenamefont {Polat}, \citenamefont {Feng}, \citenamefont {Moore}, \citenamefont {{VanderPlas}}, \citenamefont {Laxalde}, \citenamefont {Perktold}, \citenamefont {Cimrman}, \citenamefont {Henriksen}, \citenamefont {Quintero}, \citenamefont {Harris}, \citenamefont {Archibald}, \citenamefont {Ribeiro}, \citenamefont {Pedregosa}, \citenamefont {{van Mulbregt}},\ and\ \citenamefont {{SciPy 1.0 Contributors}}}]{Virtanenscipypython}%
  \BibitemOpen
  \bibfield  {author} {\bibinfo {author} {\bibfnamefont {P.}~\bibnamefont {Virtanen}}, \bibinfo {author} {\bibfnamefont {R.}~\bibnamefont {Gommers}}, \bibinfo {author} {\bibfnamefont {T.~E.}\ \bibnamefont {Oliphant}}, \bibinfo {author} {\bibfnamefont {M.}~\bibnamefont {Haberland}}, \bibinfo {author} {\bibfnamefont {T.}~\bibnamefont {Reddy}}, \bibinfo {author} {\bibfnamefont {D.}~\bibnamefont {Cournapeau}}, \bibinfo {author} {\bibfnamefont {E.}~\bibnamefont {Burovski}}, \bibinfo {author} {\bibfnamefont {P.}~\bibnamefont {Peterson}}, \bibinfo {author} {\bibfnamefont {W.}~\bibnamefont {Weckesser}}, \bibinfo {author} {\bibfnamefont {J.}~\bibnamefont {Bright}}, \bibinfo {author} {\bibfnamefont {S.~J.}\ \bibnamefont {{van der Walt}}}, \bibinfo {author} {\bibfnamefont {M.}~\bibnamefont {Brett}}, \bibinfo {author} {\bibfnamefont {J.}~\bibnamefont {Wilson}}, \bibinfo {author} {\bibfnamefont {K.~J.}\ \bibnamefont {Millman}}, \bibinfo {author} {\bibfnamefont {N.}~\bibnamefont {Mayorov}}, \bibinfo {author} {\bibfnamefont
  {A.~R.~J.}\ \bibnamefont {Nelson}}, \bibinfo {author} {\bibfnamefont {E.}~\bibnamefont {Jones}}, \bibinfo {author} {\bibfnamefont {R.}~\bibnamefont {Kern}}, \bibinfo {author} {\bibfnamefont {E.}~\bibnamefont {Larson}}, \bibinfo {author} {\bibfnamefont {C.~J.}\ \bibnamefont {Carey}}, \bibinfo {author} {\bibfnamefont {{\.I}.}~\bibnamefont {Polat}}, \bibinfo {author} {\bibfnamefont {Y.}~\bibnamefont {Feng}}, \bibinfo {author} {\bibfnamefont {E.~W.}\ \bibnamefont {Moore}}, \bibinfo {author} {\bibfnamefont {J.}~\bibnamefont {{VanderPlas}}}, \bibinfo {author} {\bibfnamefont {D.}~\bibnamefont {Laxalde}}, \bibinfo {author} {\bibfnamefont {J.}~\bibnamefont {Perktold}}, \bibinfo {author} {\bibfnamefont {R.}~\bibnamefont {Cimrman}}, \bibinfo {author} {\bibfnamefont {I.}~\bibnamefont {Henriksen}}, \bibinfo {author} {\bibfnamefont {E.~A.}\ \bibnamefont {Quintero}}, \bibinfo {author} {\bibfnamefont {C.~R.}\ \bibnamefont {Harris}}, \bibinfo {author} {\bibfnamefont {A.~M.}\ \bibnamefont {Archibald}}, \bibinfo {author}
  {\bibfnamefont {A.~H.}\ \bibnamefont {Ribeiro}}, \bibinfo {author} {\bibfnamefont {F.}~\bibnamefont {Pedregosa}}, \bibinfo {author} {\bibfnamefont {P.}~\bibnamefont {{van Mulbregt}}},\ and\ \bibinfo {author} {\bibnamefont {{SciPy 1.0 Contributors}}},\ }\bibfield  {title} {\bibinfo {title} {{{SciPy} 1.0: Fundamental Algorithms for Scientific Computing in Python}},\ }\href {https://doi.org/10.1038/s41592-019-0686-2} {\bibfield  {journal} {\bibinfo  {journal} {Nature Methods}\ }\textbf {\bibinfo {volume} {17}},\ \bibinfo {pages} {261} (\bibinfo {year} {2020})}\BibitemShut {NoStop}%
\bibitem [{\citenamefont {Biham}\ \emph {et~al.}(2002)\citenamefont {Biham}, \citenamefont {Nielsen},\ and\ \citenamefont {Osborne}}]{Bihamgrov}%
  \BibitemOpen
  \bibfield  {author} {\bibinfo {author} {\bibfnamefont {O.}~\bibnamefont {Biham}}, \bibinfo {author} {\bibfnamefont {M.~A.}\ \bibnamefont {Nielsen}},\ and\ \bibinfo {author} {\bibfnamefont {T.~J.}\ \bibnamefont {Osborne}},\ }\bibfield  {title} {\bibinfo {title} {Entanglement monotone derived from grover's algorithm},\ }\href {https://doi.org/10.1103/PhysRevA.65.062312} {\bibfield  {journal} {\bibinfo  {journal} {Phys. Rev. A}\ }\textbf {\bibinfo {volume} {65}},\ \bibinfo {pages} {062312} (\bibinfo {year} {2002})}\BibitemShut {NoStop}%
\bibitem [{\citenamefont {Hill}\ and\ \citenamefont {Wootters}(1997)}]{Hillc}%
  \BibitemOpen
  \bibfield  {author} {\bibinfo {author} {\bibfnamefont {S.~A.}\ \bibnamefont {Hill}}\ and\ \bibinfo {author} {\bibfnamefont {W.~K.}\ \bibnamefont {Wootters}},\ }\bibfield  {title} {\bibinfo {title} {Entanglement of a pair of quantum bits},\ }\href {https://doi.org/10.1103/PhysRevLett.78.5022} {\bibfield  {journal} {\bibinfo  {journal} {Phys. Rev. Lett.}\ }\textbf {\bibinfo {volume} {78}},\ \bibinfo {pages} {5022} (\bibinfo {year} {1997})}\BibitemShut {NoStop}%
\bibitem [{\citenamefont {Wootters}(1998)}]{Wootersform}%
  \BibitemOpen
  \bibfield  {author} {\bibinfo {author} {\bibfnamefont {W.~K.}\ \bibnamefont {Wootters}},\ }\bibfield  {title} {\bibinfo {title} {{Entanglement of Formation of an Arbitrary State of Two Qubits}},\ }\href {https://doi.org/10.1103/PhysRevLett.80.2245} {\bibfield  {journal} {\bibinfo  {journal} {Phys. Rev. Lett.}\ }\textbf {\bibinfo {volume} {80}},\ \bibinfo {pages} {2245} (\bibinfo {year} {1998})}\BibitemShut {NoStop}%
\bibitem [{\citenamefont {H\"ubener}\ \emph {et~al.}(2009)\citenamefont {H\"ubener}, \citenamefont {Kleinmann}, \citenamefont {Wei}, \citenamefont {Gonz\'alez-Guill\'en},\ and\ \citenamefont {G\"uhne}}]{Hubenerge}%
  \BibitemOpen
  \bibfield  {author} {\bibinfo {author} {\bibfnamefont {R.}~\bibnamefont {H\"ubener}}, \bibinfo {author} {\bibfnamefont {M.}~\bibnamefont {Kleinmann}}, \bibinfo {author} {\bibfnamefont {T.-C.}\ \bibnamefont {Wei}}, \bibinfo {author} {\bibfnamefont {C.}~\bibnamefont {Gonz\'alez-Guill\'en}},\ and\ \bibinfo {author} {\bibfnamefont {O.}~\bibnamefont {G\"uhne}},\ }\bibfield  {title} {\bibinfo {title} {Geometric measure of entanglement for symmetric states},\ }\href {https://doi.org/10.1103/PhysRevA.80.032324} {\bibfield  {journal} {\bibinfo  {journal} {Phys. Rev. A}\ }\textbf {\bibinfo {volume} {80}},\ \bibinfo {pages} {032324} (\bibinfo {year} {2009})}\BibitemShut {NoStop}%
\bibitem [{\citenamefont {Chen}\ \emph {et~al.}(2010)\citenamefont {Chen}, \citenamefont {Xu},\ and\ \citenamefont {Zhu}}]{Chenge}%
  \BibitemOpen
  \bibfield  {author} {\bibinfo {author} {\bibfnamefont {L.}~\bibnamefont {Chen}}, \bibinfo {author} {\bibfnamefont {A.}~\bibnamefont {Xu}},\ and\ \bibinfo {author} {\bibfnamefont {H.}~\bibnamefont {Zhu}},\ }\bibfield  {title} {\bibinfo {title} {Computation of the geometric measure of entanglement for pure multiqubit states},\ }\href {https://doi.org/10.1103/PhysRevA.82.032301} {\bibfield  {journal} {\bibinfo  {journal} {Phys. Rev. A}\ }\textbf {\bibinfo {volume} {82}},\ \bibinfo {pages} {032301} (\bibinfo {year} {2010})}\BibitemShut {NoStop}%
\bibitem [{\citenamefont {Tamaryan}\ \emph {et~al.}(2008)\citenamefont {Tamaryan}, \citenamefont {Park},\ and\ \citenamefont {Tamaryan}}]{Tamaryange}%
  \BibitemOpen
  \bibfield  {author} {\bibinfo {author} {\bibfnamefont {L.}~\bibnamefont {Tamaryan}}, \bibinfo {author} {\bibfnamefont {D.}~\bibnamefont {Park}},\ and\ \bibinfo {author} {\bibfnamefont {S.}~\bibnamefont {Tamaryan}},\ }\bibfield  {title} {\bibinfo {title} {Analytic expressions for geometric measure of three-qubit states},\ }\href {https://doi.org/10.1103/PhysRevA.77.022325} {\bibfield  {journal} {\bibinfo  {journal} {Phys. Rev. A}\ }\textbf {\bibinfo {volume} {77}},\ \bibinfo {pages} {022325} (\bibinfo {year} {2008})}\BibitemShut {NoStop}%
\bibitem [{\citenamefont {Shimoni}\ \emph {et~al.}(2005)\citenamefont {Shimoni}, \citenamefont {Shapira},\ and\ \citenamefont {Biham}}]{Shimoniub}%
  \BibitemOpen
  \bibfield  {author} {\bibinfo {author} {\bibfnamefont {Y.}~\bibnamefont {Shimoni}}, \bibinfo {author} {\bibfnamefont {D.}~\bibnamefont {Shapira}},\ and\ \bibinfo {author} {\bibfnamefont {O.}~\bibnamefont {Biham}},\ }\bibfield  {title} {\bibinfo {title} {{Entangled quantum states generated by Shor's factoring algorithm}},\ }\href {https://doi.org/10.1103/PhysRevA.72.062308} {\bibfield  {journal} {\bibinfo  {journal} {Phys. Rev. A}\ }\textbf {\bibinfo {volume} {72}},\ \bibinfo {pages} {062308} (\bibinfo {year} {2005})}\BibitemShut {NoStop}%
\bibitem [{\citenamefont {Most}\ \emph {et~al.}(2010)\citenamefont {Most}, \citenamefont {Shimoni},\ and\ \citenamefont {Biham}}]{Mostub}%
  \BibitemOpen
  \bibfield  {author} {\bibinfo {author} {\bibfnamefont {Y.}~\bibnamefont {Most}}, \bibinfo {author} {\bibfnamefont {Y.}~\bibnamefont {Shimoni}},\ and\ \bibinfo {author} {\bibfnamefont {O.}~\bibnamefont {Biham}},\ }\bibfield  {title} {\bibinfo {title} {{Entanglement of periodic states, the quantum Fourier transform, and Shor's factoring algorithm}},\ }\href {https://doi.org/10.1103/PhysRevA.81.052306} {\bibfield  {journal} {\bibinfo  {journal} {Phys. Rev. A}\ }\textbf {\bibinfo {volume} {81}},\ \bibinfo {pages} {052306} (\bibinfo {year} {2010})}\BibitemShut {NoStop}%
\bibitem [{\citenamefont {Skrzypczyk}\ and\ \citenamefont {Cavalcanti}(2023)}]{Skrzypczaksdpi}%
  \BibitemOpen
  \bibfield  {author} {\bibinfo {author} {\bibfnamefont {P.}~\bibnamefont {Skrzypczyk}}\ and\ \bibinfo {author} {\bibfnamefont {D.}~\bibnamefont {Cavalcanti}},\ }\href {https://doi.org/10.1088/978-0-7503-3343-6} {\emph {\bibinfo {title} {Semidefinite Programming in Quantum Information Science}}},\ 2053-2563\ (\bibinfo  {publisher} {IOP Publishing},\ \bibinfo {year} {2023})\BibitemShut {NoStop}%
\bibitem [{\citenamefont {Horodecki}(1997)}]{HORODECKIpptent}%
  \BibitemOpen
  \bibfield  {author} {\bibinfo {author} {\bibfnamefont {P.}~\bibnamefont {Horodecki}},\ }\bibfield  {title} {\bibinfo {title} {Separability criterion and inseparable mixed states with positive partial transposition},\ }\href {https://doi.org/https://doi.org/10.1016/S0375-9601(97)00416-7} {\bibfield  {journal} {\bibinfo  {journal} {Physics Letters A}\ }\textbf {\bibinfo {volume} {232}},\ \bibinfo {pages} {333} (\bibinfo {year} {1997})}\BibitemShut {NoStop}%
\bibitem [{\citenamefont {Philip}\ \emph {et~al.}(2024)\citenamefont {Philip}, \citenamefont {Rethinasamy}, \citenamefont {Russo},\ and\ \citenamefont {Wilde}}]{Philipfide}%
  \BibitemOpen
  \bibfield  {author} {\bibinfo {author} {\bibfnamefont {A.}~\bibnamefont {Philip}}, \bibinfo {author} {\bibfnamefont {S.}~\bibnamefont {Rethinasamy}}, \bibinfo {author} {\bibfnamefont {V.}~\bibnamefont {Russo}},\ and\ \bibinfo {author} {\bibfnamefont {M.~M.}\ \bibnamefont {Wilde}},\ }\bibfield  {title} {\bibinfo {title} {Schr{\"{o}}dinger as a {Q}uantum {P}rogrammer: {E}stimating {E}ntanglement via {S}teering},\ }\href {https://doi.org/10.22331/q-2024-06-11-1366} {\bibfield  {journal} {\bibinfo  {journal} {{Quantum}}\ }\textbf {\bibinfo {volume} {8}},\ \bibinfo {pages} {1366} (\bibinfo {year} {2024})}\BibitemShut {NoStop}%
\bibitem [{\citenamefont {Zhu}\ \emph {et~al.}(2024)\citenamefont {Zhu}, \citenamefont {Lv}, \citenamefont {Miao},\ and\ \citenamefont {Chen}}]{Zhumix}%
  \BibitemOpen
  \bibfield  {author} {\bibinfo {author} {\bibfnamefont {Z.}~\bibnamefont {Zhu}}, \bibinfo {author} {\bibfnamefont {G.-L.}\ \bibnamefont {Lv}}, \bibinfo {author} {\bibfnamefont {M.}~\bibnamefont {Miao}},\ and\ \bibinfo {author} {\bibfnamefont {X.-Y.}\ \bibnamefont {Chen}},\ }\bibfield  {title} {\bibinfo {title} {{Robustness of entanglement for W and Greenberger--Horne--Zeilinger mixed states}},\ }\href {https://doi.org/10.1007/s11128-024-04620-6} {\bibfield  {journal} {\bibinfo  {journal} {Quantum Information Processing}\ }\textbf {\bibinfo {volume} {24}},\ \bibinfo {pages} {5} (\bibinfo {year} {2024})}\BibitemShut {NoStop}%
\bibitem [{\citenamefont {Vidal}\ and\ \citenamefont {Tarrach}(1999)}]{Vidalrobu}%
  \BibitemOpen
  \bibfield  {author} {\bibinfo {author} {\bibfnamefont {G.}~\bibnamefont {Vidal}}\ and\ \bibinfo {author} {\bibfnamefont {R.}~\bibnamefont {Tarrach}},\ }\bibfield  {title} {\bibinfo {title} {Robustness of entanglement},\ }\href {https://doi.org/10.1103/PhysRevA.59.141} {\bibfield  {journal} {\bibinfo  {journal} {Phys. Rev. A}\ }\textbf {\bibinfo {volume} {59}},\ \bibinfo {pages} {141} (\bibinfo {year} {1999})}\BibitemShut {NoStop}%
\bibitem [{\citenamefont {Wootters}(2003)}]{woottersthermal}%
  \BibitemOpen
  \bibfield  {author} {\bibinfo {author} {\bibfnamefont {W.~K.}\ \bibnamefont {Wootters}},\ }\bibfield  {title} {\bibinfo {title} {Entangled chains},\ }\href@noop {} {\bibfield  {journal} {\bibinfo  {journal} {arXiv:quant-ph/0001114}\ } (\bibinfo {year} {2003})},\ \Eprint {https://arxiv.org/abs/quant-ph/0001114} {arXiv:quant-ph/0001114 [quant-ph]} \BibitemShut {NoStop}%
\bibitem [{\citenamefont {O'Connor}\ and\ \citenamefont {Wootters}(2001)}]{O'Connorprathermal}%
  \BibitemOpen
  \bibfield  {author} {\bibinfo {author} {\bibfnamefont {K.~M.}\ \bibnamefont {O'Connor}}\ and\ \bibinfo {author} {\bibfnamefont {W.~K.}\ \bibnamefont {Wootters}},\ }\bibfield  {title} {\bibinfo {title} {Entangled rings},\ }\href {https://doi.org/10.1103/PhysRevA.63.052302} {\bibfield  {journal} {\bibinfo  {journal} {Phys. Rev. A}\ }\textbf {\bibinfo {volume} {63}},\ \bibinfo {pages} {052302} (\bibinfo {year} {2001})}\BibitemShut {NoStop}%
\bibitem [{\citenamefont {Briegel}\ and\ \citenamefont {Raussendorf}(2001)}]{Briegelthermal}%
  \BibitemOpen
  \bibfield  {author} {\bibinfo {author} {\bibfnamefont {H.~J.}\ \bibnamefont {Briegel}}\ and\ \bibinfo {author} {\bibfnamefont {R.}~\bibnamefont {Raussendorf}},\ }\bibfield  {title} {\bibinfo {title} {Persistent entanglement in arrays of interacting particles},\ }\href {https://doi.org/10.1103/PhysRevLett.86.910} {\bibfield  {journal} {\bibinfo  {journal} {Phys. Rev. Lett.}\ }\textbf {\bibinfo {volume} {86}},\ \bibinfo {pages} {910} (\bibinfo {year} {2001})}\BibitemShut {NoStop}%
\bibitem [{\citenamefont {Arnesen}\ \emph {et~al.}(2001)\citenamefont {Arnesen}, \citenamefont {Bose},\ and\ \citenamefont {Vedral}}]{Arnesenprspinchains}%
  \BibitemOpen
  \bibfield  {author} {\bibinfo {author} {\bibfnamefont {M.~C.}\ \bibnamefont {Arnesen}}, \bibinfo {author} {\bibfnamefont {S.}~\bibnamefont {Bose}},\ and\ \bibinfo {author} {\bibfnamefont {V.}~\bibnamefont {Vedral}},\ }\bibfield  {title} {\bibinfo {title} {{Natural Thermal and Magnetic Entanglement in the 1D Heisenberg Model}},\ }\href {https://doi.org/10.1103/PhysRevLett.87.017901} {\bibfield  {journal} {\bibinfo  {journal} {Phys. Rev. Lett.}\ }\textbf {\bibinfo {volume} {87}},\ \bibinfo {pages} {017901} (\bibinfo {year} {2001})}\BibitemShut {NoStop}%
\bibitem [{\citenamefont {Wang}\ \emph {et~al.}(2001)\citenamefont {Wang}, \citenamefont {Fu},\ and\ \citenamefont {Solomon}}]{Wangspins}%
  \BibitemOpen
  \bibfield  {author} {\bibinfo {author} {\bibfnamefont {X.}~\bibnamefont {Wang}}, \bibinfo {author} {\bibfnamefont {H.}~\bibnamefont {Fu}},\ and\ \bibinfo {author} {\bibfnamefont {A.~I.}\ \bibnamefont {Solomon}},\ }\bibfield  {title} {\bibinfo {title} {{Thermal entanglement in three-qubit Heisenberg models}},\ }\href {https://doi.org/10.1088/0305-4470/34/50/312} {\bibfield  {journal} {\bibinfo  {journal} {Journal of Physics A: Mathematical and General}\ }\textbf {\bibinfo {volume} {34}},\ \bibinfo {pages} {11307} (\bibinfo {year} {2001})}\BibitemShut {NoStop}%
\bibitem [{\citenamefont {Brukner}\ and\ \citenamefont {Vedral}(2004)}]{Bruknerwitness}%
  \BibitemOpen
  \bibfield  {author} {\bibinfo {author} {\bibfnamefont {C.}~\bibnamefont {Brukner}}\ and\ \bibinfo {author} {\bibfnamefont {V.}~\bibnamefont {Vedral}},\ }\href {https://arxiv.org/abs/quant-ph/0406040} {\bibinfo {title} {Macroscopic thermodynamical witnesses of quantum entanglement}} (\bibinfo {year} {2004}),\ \Eprint {https://arxiv.org/abs/quant-ph/0406040} {arXiv:quant-ph/0406040 [quant-ph]} \BibitemShut {NoStop}%
\bibitem [{\citenamefont {\ifmmode \check{S}\else \v{S}\fi{}telmachovi\ifmmode~\check{c}\else \v{c}\fi{}}\ and\ \citenamefont {Bu\ifmmode~\check{z}\else \v{z}\fi{}ek}(2004)}]{Stelmachovicthermal}%
  \BibitemOpen
  \bibfield  {author} {\bibinfo {author} {\bibfnamefont {P.}~\bibnamefont {\ifmmode \check{S}\else \v{S}\fi{}telmachovi\ifmmode~\check{c}\else \v{c}\fi{}}}\ and\ \bibinfo {author} {\bibfnamefont {V.}~\bibnamefont {Bu\ifmmode~\check{z}\else \v{z}\fi{}ek}},\ }\bibfield  {title} {\bibinfo {title} {{Quantum-information approach to the Ising model: Entanglement in chains of qubits}},\ }\href {https://doi.org/10.1103/PhysRevA.70.032313} {\bibfield  {journal} {\bibinfo  {journal} {Phys. Rev. A}\ }\textbf {\bibinfo {volume} {70}},\ \bibinfo {pages} {032313} (\bibinfo {year} {2004})}\BibitemShut {NoStop}%
\bibitem [{\citenamefont {Gühne}\ \emph {et~al.}(2005)\citenamefont {Gühne}, \citenamefont {Tóth},\ and\ \citenamefont {Briegel}}]{Gühnespinchains}%
  \BibitemOpen
  \bibfield  {author} {\bibinfo {author} {\bibfnamefont {O.}~\bibnamefont {Gühne}}, \bibinfo {author} {\bibfnamefont {G.}~\bibnamefont {Tóth}},\ and\ \bibinfo {author} {\bibfnamefont {H.~J.}\ \bibnamefont {Briegel}},\ }\bibfield  {title} {\bibinfo {title} {Multipartite entanglement in spin chains},\ }\href {https://doi.org/10.1088/1367-2630/7/1/229} {\bibfield  {journal} {\bibinfo  {journal} {New Journal of Physics}\ }\textbf {\bibinfo {volume} {7}},\ \bibinfo {pages} {229} (\bibinfo {year} {2005})}\BibitemShut {NoStop}%
\bibitem [{\citenamefont {Amico}\ \emph {et~al.}(2008)\citenamefont {Amico}, \citenamefont {Fazio}, \citenamefont {Osterloh},\ and\ \citenamefont {Vedral}}]{Amicothermalreviewentang}%
  \BibitemOpen
  \bibfield  {author} {\bibinfo {author} {\bibfnamefont {L.}~\bibnamefont {Amico}}, \bibinfo {author} {\bibfnamefont {R.}~\bibnamefont {Fazio}}, \bibinfo {author} {\bibfnamefont {A.}~\bibnamefont {Osterloh}},\ and\ \bibinfo {author} {\bibfnamefont {V.}~\bibnamefont {Vedral}},\ }\bibfield  {title} {\bibinfo {title} {Entanglement in many-body systems},\ }\href {https://doi.org/10.1103/RevModPhys.80.517} {\bibfield  {journal} {\bibinfo  {journal} {Rev. Mod. Phys.}\ }\textbf {\bibinfo {volume} {80}},\ \bibinfo {pages} {517} (\bibinfo {year} {2008})}\BibitemShut {NoStop}%
\bibitem [{\citenamefont {Breunig}\ \emph {et~al.}(2013)\citenamefont {Breunig}, \citenamefont {Garst}, \citenamefont {Sela}, \citenamefont {Buldmann}, \citenamefont {Becker}, \citenamefont {Bohat\'y}, \citenamefont {M\"uller},\ and\ \citenamefont {Lorenz}}]{Breunigspinchains}%
  \BibitemOpen
  \bibfield  {author} {\bibinfo {author} {\bibfnamefont {O.}~\bibnamefont {Breunig}}, \bibinfo {author} {\bibfnamefont {M.}~\bibnamefont {Garst}}, \bibinfo {author} {\bibfnamefont {E.}~\bibnamefont {Sela}}, \bibinfo {author} {\bibfnamefont {B.}~\bibnamefont {Buldmann}}, \bibinfo {author} {\bibfnamefont {P.}~\bibnamefont {Becker}}, \bibinfo {author} {\bibfnamefont {L.}~\bibnamefont {Bohat\'y}}, \bibinfo {author} {\bibfnamefont {R.}~\bibnamefont {M\"uller}},\ and\ \bibinfo {author} {\bibfnamefont {T.}~\bibnamefont {Lorenz}},\ }\bibfield  {title} {\bibinfo {title} {{Spin-$\frac{1}{2}$ $XXZ$ Chain System ${\mathrm{Cs}}_{2}{\mathrm{CoCl}}_{4}$ in a Transverse Magnetic Field}},\ }\href {https://doi.org/10.1103/PhysRevLett.111.187202} {\bibfield  {journal} {\bibinfo  {journal} {Phys. Rev. Lett.}\ }\textbf {\bibinfo {volume} {111}},\ \bibinfo {pages} {187202} (\bibinfo {year} {2013})}\BibitemShut {NoStop}%
\bibitem [{\citenamefont {T\'oth}\ \emph {et~al.}(2009)\citenamefont {T\'oth}, \citenamefont {Knapp}, \citenamefont {G\"uhne},\ and\ \citenamefont {Briegel}}]{Tothboundentanglementinspinchains}%
  \BibitemOpen
  \bibfield  {author} {\bibinfo {author} {\bibfnamefont {G.}~\bibnamefont {T\'oth}}, \bibinfo {author} {\bibfnamefont {C.}~\bibnamefont {Knapp}}, \bibinfo {author} {\bibfnamefont {O.}~\bibnamefont {G\"uhne}},\ and\ \bibinfo {author} {\bibfnamefont {H.~J.}\ \bibnamefont {Briegel}},\ }\bibfield  {title} {\bibinfo {title} {Spin squeezing and entanglement},\ }\href {https://doi.org/10.1103/PhysRevA.79.042334} {\bibfield  {journal} {\bibinfo  {journal} {Phys. Rev. A}\ }\textbf {\bibinfo {volume} {79}},\ \bibinfo {pages} {042334} (\bibinfo {year} {2009})}\BibitemShut {NoStop}%
\bibitem [{\citenamefont {Sachdev}(2011)}]{Sachdevqft}%
  \BibitemOpen
  \bibfield  {author} {\bibinfo {author} {\bibfnamefont {S.}~\bibnamefont {Sachdev}},\ }\href {https://doi.org/https://doi.org/10.1017/CBO9780511973765} {\emph {\bibinfo {title} {Quantum Phase Transitions}}},\ \bibinfo {edition} {2nd}\ ed.\ (\bibinfo  {publisher} {Cambridge University Press},\ \bibinfo {year} {2011})\BibitemShut {NoStop}%
\bibitem [{\citenamefont {Osterloh}\ \emph {et~al.}(2002)\citenamefont {Osterloh}, \citenamefont {Amico}, \citenamefont {Falci},\ and\ \citenamefont {Fazio}}]{Osterlohqft}%
  \BibitemOpen
  \bibfield  {author} {\bibinfo {author} {\bibfnamefont {A.}~\bibnamefont {Osterloh}}, \bibinfo {author} {\bibfnamefont {L.}~\bibnamefont {Amico}}, \bibinfo {author} {\bibfnamefont {G.}~\bibnamefont {Falci}},\ and\ \bibinfo {author} {\bibfnamefont {R.}~\bibnamefont {Fazio}},\ }\bibfield  {title} {\bibinfo {title} {Scaling of entanglement close to a quantum phase transition},\ }\href {https://doi.org/10.1038/416608a} {\bibfield  {journal} {\bibinfo  {journal} {Nature}\ }\textbf {\bibinfo {volume} {416}},\ \bibinfo {pages} {608} (\bibinfo {year} {2002})}\BibitemShut {NoStop}%
\bibitem [{\citenamefont {Gu}\ \emph {et~al.}(2003)\citenamefont {Gu}, \citenamefont {Lin},\ and\ \citenamefont {Li}}]{Guqft}%
  \BibitemOpen
  \bibfield  {author} {\bibinfo {author} {\bibfnamefont {S.-J.}\ \bibnamefont {Gu}}, \bibinfo {author} {\bibfnamefont {H.-Q.}\ \bibnamefont {Lin}},\ and\ \bibinfo {author} {\bibfnamefont {Y.-Q.}\ \bibnamefont {Li}},\ }\bibfield  {title} {\bibinfo {title} {Entanglement, quantum phase transition, and scaling in the $\mathrm{XXZ}$ chain},\ }\href {https://doi.org/10.1103/PhysRevA.68.042330} {\bibfield  {journal} {\bibinfo  {journal} {Phys. Rev. A}\ }\textbf {\bibinfo {volume} {68}},\ \bibinfo {pages} {042330} (\bibinfo {year} {2003})}\BibitemShut {NoStop}%
\bibitem [{\citenamefont {Glaser}\ \emph {et~al.}(2003)\citenamefont {Glaser}, \citenamefont {B\"uttner},\ and\ \citenamefont {Fehske}}]{Glaserqft}%
  \BibitemOpen
  \bibfield  {author} {\bibinfo {author} {\bibfnamefont {U.}~\bibnamefont {Glaser}}, \bibinfo {author} {\bibfnamefont {H.}~\bibnamefont {B\"uttner}},\ and\ \bibinfo {author} {\bibfnamefont {H.}~\bibnamefont {Fehske}},\ }\bibfield  {title} {\bibinfo {title} {Entanglement and correlation in anisotropic quantum spin systems},\ }\href {https://doi.org/10.1103/PhysRevA.68.032318} {\bibfield  {journal} {\bibinfo  {journal} {Phys. Rev. A}\ }\textbf {\bibinfo {volume} {68}},\ \bibinfo {pages} {032318} (\bibinfo {year} {2003})}\BibitemShut {NoStop}%
\bibitem [{\citenamefont {Carvalho}\ \emph {et~al.}(2004)\citenamefont {Carvalho}, \citenamefont {Mintert},\ and\ \citenamefont {Buchleitner}}]{Carvalhoamplitude}%
  \BibitemOpen
  \bibfield  {author} {\bibinfo {author} {\bibfnamefont {A.~R.~R.}\ \bibnamefont {Carvalho}}, \bibinfo {author} {\bibfnamefont {F.}~\bibnamefont {Mintert}},\ and\ \bibinfo {author} {\bibfnamefont {A.}~\bibnamefont {Buchleitner}},\ }\bibfield  {title} {\bibinfo {title} {Decoherence and multipartite entanglement},\ }\href {https://doi.org/10.1103/PhysRevLett.93.230501} {\bibfield  {journal} {\bibinfo  {journal} {Phys. Rev. Lett.}\ }\textbf {\bibinfo {volume} {93}},\ \bibinfo {pages} {230501} (\bibinfo {year} {2004})}\BibitemShut {NoStop}%
\bibitem [{\citenamefont {Ali}\ and\ \citenamefont {Gühne}(2014)}]{Alinoi}%
  \BibitemOpen
  \bibfield  {author} {\bibinfo {author} {\bibfnamefont {M.}~\bibnamefont {Ali}}\ and\ \bibinfo {author} {\bibfnamefont {O.}~\bibnamefont {Gühne}},\ }\bibfield  {title} {\bibinfo {title} {Robustness of multiparticle entanglement: specific entanglement classes and random states},\ }\href {https://doi.org/10.1088/0953-4075/47/5/055503} {\bibfield  {journal} {\bibinfo  {journal} {Journal of Physics B: Atomic, Molecular and Optical Physics}\ }\textbf {\bibinfo {volume} {47}},\ \bibinfo {pages} {055503} (\bibinfo {year} {2014})}\BibitemShut {NoStop}%
\end{thebibliography}%
\onecolumngrid
\appendix

\section{Pure states}{\label{se::pu}}
Entanglement in a pure $(M+1)$-partite state can be upper bounded by:
\begin{equation}
    E_G(\ket{\psi}^{A_1\ldots A_MB}) \geq 1 - \max_{\sigma^{A_1\ldots A_M} \in \mathrm{PPT}} \Tr(\rho^{A_1\ldots A_M} \sigma^{A_1\ldots A_M}).
    \label{eq::lopua}
\end{equation}
It often happens that $\sigma^{A_1\ldots A_M}$ that maximizes above expression is a pure state. We will derive a bound on accuracy of estimation of entanglement using \eqref{eq::lopua}.
Let $\psi_{A_{[1,M]}}$ be a multipartite pure state, where $M\in\mathbb{N}$ and $[n,M]=\{n,n+1,\ldots,M\}$ $\forall\, n\leq M$. Further $\sigma_{\operatorname{PPT}}$ is multipartite state which is PPT in all bipartitions such that $F(\psi,\sigma_{\operatorname{PPT}}) = 1-\varepsilon$ with some $\varepsilon > 0$. Using the same
arguments as in the proof of Eq. (5) in \cite{Ganardilocalpurity}, there
exist product states 
\begin{align}
     F(\psi_{A_{[1,M]}}, \phi_{A_{[1,i+1]}}\otimes \phi_{A_{[i+1,M]}}) &> 1-\varepsilon\\
     F(\psi_{A_{[1,M]}}, \phi_{A_{[2,i+1]}}\otimes \phi_{A_{1}A_{[i+2,M]}}) &> 1-\varepsilon
\end{align}
Then, using the triangle inequality of trace distance, Fuchs-van de Graaf inequality, we get
\begin{align}
    &\Vert \psi_{A_{[1,M]}} - \bigotimes_{i=1}^{M}\phi_{A_{i}}\Vert_{1} \leq \Vert \psi_{A_{[1,M]}} - \phi_{A_{1}}\otimes\phi_{A_{[2,M]}}\Vert_{1} + \Vert \phi_{A_{1}}\otimes\phi_{A_{[2,M]}} - \bigotimes_{i=1}^{M}\phi_{A_{i}}\Vert_{1}\\\nonumber
    &=\Vert \psi_{A_{[M]}} - \phi_{A_{1}}\otimes\phi_{A_{[2,M]}}\Vert_{1} + \Vert \phi_{A_{[2,M]}} - \bigotimes_{i=2}^{M}\phi_{A_{i}}\Vert_{1}\\\nonumber
    &\leq \Vert \psi_{A_{[M]}} - \phi_{A_{1}}\otimes\phi_{A_{[2,M]}}\Vert_{1} + \Vert \phi_{A_{[2,M]}} -\phi_{A_{2}}\otimes\phi_{A_{[3,M]}}\Vert_{1} +\Vert \phi_{A_{2}}\otimes\phi_{A_{[3,M]}}- \bigotimes_{i=2}^{M}\phi_{A_{i}}\Vert_{1}\\\nonumber
    &= \Vert \psi_{A_{[M]}} - \phi_{A_{1}}\otimes\phi_{A_{[2,M]}}\Vert_{1} + \Vert \phi_{A_{[2,M]}} -\phi_{A_{2}}\otimes\phi_{A_{[3,M]}}\Vert_{1} +\Vert \phi_{A_{[3,M]}}- \bigotimes_{i=3}^{M}\phi_{A_{i}}\Vert_{1}
\end{align}
And so on until,
\begin{align}
    &\Vert \psi_{A_{[1,M]}} - \bigotimes_{i=1}^{M}\phi_{A_{i}}\Vert_{1}\leq  \Vert \psi_{A_{[M]}} - \phi_{A_{1}}\otimes\phi_{A_{[2,M]}}\Vert_{1} + \sum_{i=1}^{M-2}\Vert \phi_{A_{[i+1,M]}} -\phi_{A_{i+1}}\otimes\phi_{A_{[i+2,M]}}\Vert_{1}\\\nonumber
    &\leq  \Vert \psi_{A_{[M]}} - \phi_{A_{1}}\otimes\phi_{A_{[2,M]}}\Vert_{1} + \sum_{i=1}^{M-2}\Vert \phi_{A_{[1,i+1]}}\otimes \phi_{A_{[i+1,M]}} -\phi_{A_{[2,i+1]}}\otimes\phi_{A_{1}A_{[i+2,M]}}\Vert_{1}\\\nonumber
    &\leq \Vert \psi_{A_{[M]}} - \phi_{A_{1}}\otimes\phi_{A_{[2,M]}}\Vert_{1}\\\nonumber
    &\qquad\qquad\qquad + \sum_{i=1}^{M-2}\Vert \phi_{A_{[1,i+1]}}\otimes \phi_{A_{[i+1,M]}}-\psi_{A_{[M]}}\Vert_{1}+\Vert\psi_{A_{[M]}} -\phi_{A_{[2,i+1]}}\otimes\phi_{A_{1}A_{[i+2,M]}}\Vert_{1}\\\nonumber
    &\leq (4M-6)\sqrt{\varepsilon}
\end{align}
This implies that there exists a separable state $\sigma_{\operatorname{SEP}}$ such that 
\begin{equation}
    \Vert\sigma_{\operatorname{SEP}}-\psi\Vert_{1} \leq (4M-6)\sqrt{\varepsilon}
\end{equation}
Using the arguments in \cite[Appendix B]{Ganardilocalpurity}, we get that
\begin{align}
    \vert\operatorname{Tr}(\rho\sigma_{\operatorname{PPT}}) - \operatorname{Tr}(\rho\sigma_{\operatorname{SEP}})\vert &\leq \Vert\sigma_{\operatorname{PPT}}-\sigma_{\operatorname{SEP}}\Vert_{1}\\\nonumber
    &\leq \Vert\sigma_{\operatorname{PPT}}-\psi\Vert_{1}+\Vert\sigma_{\operatorname{SEP}}-\psi\Vert_{1}\\\nonumber
    &\leq 4(M-1)\sqrt{\varepsilon}
\end{align}
Putting it all together, we obtain that for an (M)-partite pure state, the accuracy of the estimation of the geometric entanglement by \eqref{eq::lopua} is given by $4(M-2)\sqrt{\varepsilon}.$
\section{Proof of Theorem~\ref{thm:pec_mult}}{\label{se::pr}}
    \begin{definition}
        Multipartite Geometric Entanglement for a pure state $\psi$ is defined as follows:
        \begin{equation}\label{eqn:geometric_entanglement}
            E_{g,M}(\psi^{B_{1}\cdots B_{M}}) = 1 - \max_{\phi\in\operatorname{M-SEP}}\vert\langle\phi\vert\psi^{B_{1}\ldots B_{M}}\rangle\vert^2
        \end{equation}
        where $\operatorname{M-SEP}$ is the state of $M$-partite product states. 
    \end{definition}
    Recall, the following fact about $\max_{\phi\in\operatorname{M-SEP}}\vert\langle\phi\vert\psi\rangle\vert^2$ from~\cite{Streltsov_2010}:
    \begin{equation}\label{eqn:geometric_entanglement_1}
        \max_{\phi\in\operatorname{M-SEP}}\vert\langle\phi\vert\psi^{B_{1}\ldots B_{M}}\rangle\vert^2 = \max_{\vert\phi_i\rangle^{B_{i}}}\operatorname{Tr}\left[\left(\vert\phi_{1}\rangle\langle\phi_{1}\vert^{B_{1}}\otimes\cdots\vert\phi_{M-1}\rangle\langle\phi_{M-1}\vert^{B_{M-1}}\right)\psi^{B_{1}\ldots B_{M-1}}\right]
    \end{equation}
    \begin{definition}
        Multipartite Geometric Entanglement for a mixed state $\rho$ is defined as follows:
        \begin{equation}\label{eqn:geometric_entanglement}
            E_{g,M}(\rho) = \min \sum_{i}p_{i}E_{g}(\psi_{i})
        \end{equation}
        where minimization is taken over all pure state decompositions of $\rho$ such that $\sum_{i}p_{i}\psi_{i}$.
    \end{definition}
    
    Following the analysis from \cite{Ganardilocalpurity}, consider the following scenario: Alice $A$ teams up with the many Bobs $B_{1}\otimes\cdots\otimes B_{M-1}$ to cool his individual system to the ground state. The figure of merit is then given by
        \begin{equation}
            F^{A\vert B_{1}\cdots B_{M-1}}(\rho^{AB_{1}\cdots B_{M-1}}) = \max_{\Lambda\in\operatorname{GLOCC}}\langle 0\vert^{B_{1}}\cdots \langle 0\vert^{B_{M-1}}\operatorname{Tr}_{A}\left[\Lambda(\rho^{AB_{1}\cdots B_{M}})\right] \vert 0\rangle^{B_{M-1}}\cdots\vert 0\rangle^{B_{1}} .
        \end{equation}
        where GLOCC is the set of Gibbs preserving LOCC channels. We can rewrite the above equation in the following fashion:
    \begin{equation}
        F^{A\vert B_{1}\cdots B_{M-1}}(\rho^{AB_{1}\cdots B_{M-1}}) = \max_{\Lambda\in\operatorname{GLOCC}}\operatorname{Tr}\left[(\openone^{A}\otimes\vert 0\rangle\langle 0\vert^{B_{1}}\otimes\cdots\otimes\vert 0\rangle\langle 0\vert^{B_{M-1}})\,\Lambda(\rho^{AB_{1}\cdots B_{M-1}})\right].
    \end{equation}
    Let us further restrict the set of allowed operation such that classical communication flows exclusively from Alice to Bob. When Alice executes a local measurement characterized by POVM elements $\{M^{A}_{i}\}$, the probability of obtaining the measurement outcome $i$ manifests as: $p_{i} = \operatorname{Tr}[(M^{A}_{i}\otimes\openone^{B_{1}\cdots B_{M}})\rho^{AB_{1}\cdots B_{M}}]$. Given Alice's measurement outcome, the resultant state of Bobs' systems is $\sigma^{B_{1}\cdots B_{M}}_{i}=  \operatorname{Tr}_{A}[(M^{A}_{i}\otimes\openone^{B_{1}\cdots B_{M}})\rho^{AB_{1}\cdots B_{M}}]/p_{i}$ The maximal fidelity Bob can achieve with the ground state in this setup is represented as:
    \begin{align}
        F^{A\vert B_{1}\cdots B_{M-1}}_{\rightarrow}&(\rho^{AB_{1}\cdots B_{M-1}})\notag\\
        &= \max_{\Lambda\in\operatorname{GLOCC}}\sum_{i}\operatorname{Tr}\left[\Lambda^\dagger(\openone^{A}\otimes\vert 0\rangle\langle 0\vert^{B_{1}}\otimes\cdots\otimes\vert 0\rangle\langle 0\vert^{B_{M-1}})(\rho^{AB_{1}\cdots B_{M-1}})\right]\nonumber\\
        &= \max_{\{M^{A}_{i}\},\{\mu_{i,k}^{B_{k}}\}}\sum_{i}\operatorname{Tr}\left[\left(M^{A}_{i}\otimes\mu_{i,1}^{B_{1}}\otimes\cdots\otimes\mu_{i,M-1}^{B_{M-1}}\right)(\rho^{AB_{1}\cdots B_{M-1}})\right]\label{eqn:basis_4_sdp}\\
        &= \max_{\{M^{A}_{i}\}}\sum_{i}p_{i}\max_{\{\mu_{i,k}^{B_{k}}\}}\operatorname{Tr}\left[(\mu_{i,1}^{B_{1}}\otimes\cdots\otimes\mu_{i,M-1}^{B_{M-1}})(\sigma^{B_{1}\cdots B_{M-1}})\right]\nonumber
    \end{align}
    where the maximization is over all POVMs $\{M^{A}_{i}\}$ on $A$ and  $\{\mu_{i,k}^{B_{k}}\}$ are pure states $B_{k}$ for all $i$ and $k$.
    
    Using \eqref{eqn:basis_4_sdp}, we can see that $F^{A\vert B_{1}\cdots B_{M-1}}_{\rightarrow}(\rho^{AB_{1}\cdots B_{M-1}})$ is upper bounded by the following semi-definite program:
    \begin{equation}
        F^{A\vert B_{1}\cdots B_{M-1}}_{\rightarrow}(\rho^{AB_{1}\cdots B_{M-1}})\leq \max_{X^{AB_{1}\cdots B_{M-1}}}\operatorname{Tr}\left[X^{AB_{1}\cdots B_{M-1}}\rho^{A B_{1}\cdots B_{M-1}}\right]
    \end{equation}
    where the optimization is over all matrices $X^{AB_{1}\cdots B_{M-1}}$ such that 
    \begin{equation}
        X^{AB_{1}\cdots B_{M-1}}\geq 0,\quad X^{A} = \openone^{A},\quad X^{T_{\tilde{B}}}\geq 0\,\, \forall \tilde{B}
    \end{equation}
    where $\tilde{B}$ is any bipartition of $B_{1}\cdots B_{M-1}$.\\
    
    Now, we can state and prove the following theorem.
    \begin{thm}
        For any multipartite pure state $\psi^{AB_{1}\cdots B_{M}}$, the following equation holds:
        \begin{equation}
            F^{A\vert B_{1}\cdots B_{M-1}}_{\rightarrow}(\rho^{AB_{1}\cdots B_{M-1}}) + E_{g,M}(\rho^{B_{1}\cdots B_{M}}) = 1
        \end{equation}
    \end{thm}
    \begin{proof}
        Let $\{M_{i}^{A}\}$ be the POVM applied by Alice to achieve the maximum in \eqref{eqn:basis_4_sdp}. Then, the probability $p_{i}$ of obtaining the measurement outcome $i$ and the post measurement state for the Bobs are given as follows:
        \begin{align}
            p_{i} &= \operatorname{Tr}\left[\left(M^{A}_{i}\otimes\openone^{B_{1}}\otimes\cdots\otimes\openone^{B_{M-1}}\right)(\rho^{AB_{1}\cdots B_{M-1}})\right]\\\nonumber
            \sigma_{i}^{B_{1}\cdots B_{M-1}} &= \frac{1}{p_{i}}\operatorname{Tr}_{A}\left[\left(M^{A}_{i}\otimes\openone^{B_{1}}\otimes\cdots\otimes\openone^{B_{M-1}}\right)(\rho^{AB_{1}\cdots B_{M-1}})\right]
        \end{align}
        We then get that, 
        \begin{align}\label{eqn:f_1w_simplified}
            F^{A\vert B_{1}\cdots B_{M-1}}_{\rightarrow}&(\rho^{AB_{1}\cdots B_{M-1}})= \max_{\{\mu_{i,k}^{B_{k}}\}}\sum_{i}\operatorname{Tr}\left[\left(M^{A}_{i}\otimes\mu_{i,1}^{B_{1}}\otimes\cdots\otimes\mu_{i,M-1}^{B_{M-1}}\right)(\rho^{AB_{1}\cdots B_{M-1}})\right].
        \end{align}
        Without loss of generality, we can assume that elements of the POVM $\{M_{i}^{A}\}$ are rank one.
        
        Now consider the action of $\{M_{i}^{A}\}$ on the pure state, $\psi^{AB_{1}\cdots B_{M}}$ which is the purification of $\rho^{AB_{1}\cdots B_{M-1}}$. Given that the measurement outcome obtained is $i$, the post measurement state is 
        \begin{equation}
            \sigma_{i}^{B_{1}\cdots B_{M}} = \frac{1}{p_{i}}\operatorname{Tr}_{A}\left[\left(M^{A}_{i}\otimes\openone^{B_{1}}\otimes\cdots\otimes\openone^{B_{M}}\right)(\psi^{AB_{1}\cdots B_{M}})\right]
        \end{equation}
        We can calculate the geometric entanglement of $\sigma_{i}^{B_{1}\cdots B_{M}}$ using \eqref{eqn:geometric_entanglement_1}, to get
        \begin{equation}
            E_{g,M}(\sigma_{i}^{B_{1}\cdots B_{M}}) = 1 - \frac{1}{p_{i}}\max_{\{\mu_{i,k}^{B_{k}}\}}\operatorname{Tr}\left[\left(M^{A}_{i}\otimes\mu_{i,1}^{B_{1}}\otimes\cdots\otimes\mu_{i,M-1}^{B_{M-1}}\right)(\psi^{AB_{1}\cdots B_{M-1}})\right]
        \end{equation}
        Together with \eqref{eqn:f_1w_simplified}, we get
        \begin{equation}
            F^{A\vert B_{1}\cdots B_{M-1}}_{\rightarrow}(\rho^{AB_{1}\cdots B_{M-1}}) = 1- \sum_{i}p_{i}E_{g,M}(\sigma_{i}^{B_{1}\cdots B_{M}})
        \end{equation}
        Using convexity of the geometric entanglement, we can further say that
        \begin{equation}
            F^{A\vert B_{1}\cdots B_{M-1}}_{\rightarrow}(\rho^{AB_{1}\cdots B_{M-1}}) \leq 1- E_{g,M}(\rho^{B_{1}\cdots B_{M}})
        \end{equation}
        Note that there always exists a pure state decomposition $\{q_{i},\phi_{i}^{B_{1}\cdots B_{M}}\}$ of $\rho^{B_{1}\cdots B_{M}}$ such that 
        \begin{equation}
            E_{g,M}(\rho^{B_{1}\cdots B_{M}}) = 1- \sum_{i}q_{i}E_{g,M}(\phi_{i}^{B_{1}\cdots B_{M}}) = 1 - \max_{\{\mu_{i,k}^{B_{k}}\}}\operatorname{Tr}\left[\left(\mu_{i,1}^{B_{1}}\otimes\cdots\otimes\mu_{i,M-1}^{B_{M-1}}\right)(\phi_{i}^{B_{1}\cdots B_{M-1}})\right]
        \end{equation}
        Also, there exists a POVM $\{M_{i}^{A}\}$ that can be applied on $\psi^{AB_{1}\cdots B_{M}}$, the purification  of $\rho^{B_{1}\cdots B_{M}}$, to obtain the ensemble $\{q_{i},\phi_{i}^{B_{1}\cdots B_{M}}\}$. Furthermore, for the $\rho^{AB_{1}\cdots B_{M-1}}$, the maximal fidelity to the ground state is shown to be   equal to $1- E_{g,M}(\rho^{B_{1}\cdots B_{M}})$ in~\eqref{eqn:basis_4_sdp}. This concludes the proof. 
    \end{proof} 

    Similar to the theorems in \cite{Ganardilocalpurity}, we can see that $E_{g,M}(\rho^{B_{1}\cdots B_{M}})$ is lower bounded by the following semi-definite program:
    \begin{equation}
        E_{g,M}(\rho^{B_{1}\cdots B_{M}})\geq 1- \max_{X^{AB_{1}\cdots B_{M}}}\operatorname{Tr}\left[X^{AB_{1}\cdots B_{M}}\psi^{A B_{1}\cdots B_{M}}\right]
    \end{equation}
    where $\psi^{A B_{1}\cdots B_{M}}$ is a purification of $\rho^{B_{1}\cdots B_{M}}$, and the optimization is over all matrices $X^{AB_{1}\cdots B_{M}}$ such that 
    \begin{equation}
        X^{AB_{1}\cdots B_{M}}\geq 0,\quad X^{A} = \openone^{A},\quad X^{T_{\tilde{B}}}\geq 0\,\, \forall \tilde{B}
    \end{equation}
    where $\tilde{B}$ is any bipartition of $B_{1}\cdots B_{M}$.

\section{How bounds look in multipartite scenario}{\label{se::multi}}

The crucial part of our method is performing the PPT relaxation. 
Assume we have a multipartite state \(\rho^{A_1A_2\ldots A_n}\). 
If the state is fully separable, it satisfies separability conditions across all possible partitions. 
We relax these conditions in the following way: we perform the PPT relaxation by imposing the PPT constraint in each bipartition \(A_i\) versus the remaining subsystems. 
It is also possible to tighten the relaxation by imposing the PPT criterion on all bipartitions; however, numerical computations show that this does not improve the accuracy while significantly increasing the computational cost. For example, consider the modified bound \eqref{eq::bopurc} for a tripartite state.\\

\textbf{Lower bound 3'}
    We can bound the geometric entanglement from below by the following SDP.
    \begin{equation}
        E_{g}(\rho^{B_1\ldots B_n})\geq1-\max_{X^{AB_1\ldots B_{n-1}}} \;  \, \mathrm{Tr} \left( X^{AB_1\ldots B_{n-1}}\rho^{AB_1\ldots B_{n-1}}\right),
    \end{equation}
    where $\rho^{AB_1\ldots B_{n}}$ is the purification of $\rho^{B_1\ldots B_n}.$
    Subject to the following constraints:
    \[
\begin{aligned}
& X^{AB_1\ldots B_{n-1}} \succeq 0 \quad  
 X^{A} = \openone^{A} \quad \\
 & X^{AB_1\ldots B_{n-1}} \succeq 0 \quad where\quad  i=A,B_1,\ldots ,B_{n-1}. 
\end{aligned} 
\]
Other bounds are adjusted in similar manner.
Note that in Appendix \ref{se::pr} we demanded $X$ to be PPT in all bipartitions. Here, to accelerate computations, we relax some of the  constraints and only demand $X$ to be PPT when we transpose each subsystem. By doing this, we are computing a lower bound to our proposed lower bound.

\section{K-symmetric extension in the multipartite scenario}{\label{se::k}}

In the bipartite case, a k-symmetric extension of a state $\rho^{AB_1}$ is defined by: $\rho^{AB_1\ldots B_k},$ with additional condition
\begin{equation}
    \rho^{AB_1}=\rho^{AB_k}.
    \label{eq::ksymrho}
\end{equation}
For multipartite state $\rho^{ABC},$ we define a k-symmetric extension by adding k additional copies of the last 2  subsystems. For example, for $k=3,$ we obtain the state $\rho^{AB_1C_1B_2C_2C_3},$ with the following constraints:
\begin{equation}
    \rho^{AB_1C_1}=\rho^{AB_1C_i},
    \label{eq::k3}
\end{equation}
where $i=2,3,$ and 
\begin{equation}
    \rho^{AB_1C_1}=\rho^{AB_2C_2}.
    \label{eq::k3addit}
\end{equation}
For greater k and higher partite states the k-symmetric extension is build similarly.

\section{Phase transition in XXX model}{\label{s::pha}}
The thermal state of Hamiltonian \eqref{eq::xxx} is given by: 
\begin{eqnarray}
    \rho_{m}=&\left(
\begin{array}{cccccccc}
 \frac{1}{\mu } & 0 & 0 & 0 & 0 & 0 & 0 & 0 \\
 0 & \frac{\xi  e^{2 \beta  h}}{3 \mu } & -\frac{\kappa  e^{2 \beta  h}}{3 \mu } & 0 & -\frac{\kappa  e^{2 \beta  h}}{3 \mu } & 0 & 0 & 0 \\
 0 & -\frac{\kappa  e^{2 \beta  h}}{3 \mu } & \frac{\xi  e^{2 \beta  h}}{3 \mu } & 0 & -\frac{\kappa  e^{2 \beta  h}}{3 \mu } & 0 & 0 & 0 \\
 0 & 0 & 0 & \frac{\xi  e^{4 \beta  h}}{3 \mu } & 0 & -\frac{\kappa  e^{4 \beta  h}}{3 \mu } & -\frac{\kappa  e^{4 \beta  h}}{3 \mu } & 0 \\
 0 & -\frac{\kappa  e^{2 \beta  h}}{3 \mu } & -\frac{\kappa  e^{2 \beta  h}}{3 \mu } & 0 & \frac{\xi  e^{2 \beta  h}}{3 \mu } & 0 & 0 & 0 \\
 0 & 0 & 0 & -\frac{\kappa  e^{4 \beta  h}}{3 \mu } & 0 & \frac{\xi  e^{4 \beta  h}}{3 \mu } & -\frac{\kappa  e^{4 \beta  h}}{3 \mu } & 0 \\
 0 & 0 & 0 & -\frac{\kappa  e^{4 \beta  h}}{3 \mu } & 0 & -\frac{\kappa  e^{4 \beta  h}}{3 \mu } & \frac{\xi  e^{4 \beta  h}}{3 \mu } & 0 \\
 0 & 0 & 0 & 0 & 0 & 0 & 0 & \frac{e^{6 \beta  h}}{\mu } \\
\end{array}
\right),
\label{magneticthermalstate}
\end{eqnarray}
where $\kappa=e^{3 \beta  J}-1,$ $\xi=2 e^{3 \beta  J}+1,$ $\mu=\left(e^{2 \beta  h}+1\right) \left(e^{4 \beta  h}+2 e^{2 \beta  h+3 \beta  J}+1\right).$ Consider the ground state of \eqref{magneticthermalstate}, i.e. state \eqref{magneticthermalstate} in the limit $\beta\rightarrow\infty.$ When we change magnetic field $h,$ we observe discontinuous change for $h=\frac{3}{2}J,$ $h=0$ and $h=-\frac{3}{2}J.$ This points corresponds to phase transition points we observe in the Fig. \ref{fig::mag} in the main part of the paper. For $h>\frac{3}{2}J,$ the ground state is $\ket{111}\bra{111}.$ It is a separable state. For $h=\frac{3}{2}J,$ the ground state is 
\begin{equation}
    \left(
\begin{array}{cccccccc}
 0 & 0 & 0 & 0 & 0 & 0 & 0 & 0 \\
 0 & 0 & 0 & 0 & 0 & 0 & 0 & 0 \\
 0 & 0 & 0 & 0 & 0 & 0 & 0 & 0 \\
 0 & 0 & 0 & \frac{2}{9} & 0 & -\frac{1}{9} & -\frac{1}{9} & 0 \\
 0 & 0 & 0 & 0 & 0 & 0 & 0 & 0 \\
 0 & 0 & 0 & -\frac{1}{9} & 0 & \frac{2}{9} & -\frac{1}{9} & 0 \\
 0 & 0 & 0 & -\frac{1}{9} & 0 & -\frac{1}{9} & \frac{2}{9} & 0 \\
 0 & 0 & 0 & 0 & 0 & 0 & 0 & \frac{1}{3} \\
\end{array}
\right).
\label{m32}
\end{equation}
The geometric entanglement of the above state is $0.115599055$ with accuracy $3\times10^{-7}.$ 
The ground state for $0<h<\frac{3}{2}$ is:
\begin{equation}
    \left(
\begin{array}{cccccccc}
 0 & 0 & 0 & 0 & 0 & 0 & 0 & 0 \\
 0 & 0 & 0 & 0 & 0 & 0 & 0 & 0 \\
 0 & 0 & 0 & 0 & 0 & 0 & 0 & 0 \\
 0 & 0 & 0 & \frac{1}{3} & 0 & -\frac{1}{6} & -\frac{1}{6} & 0 \\
 0 & 0 & 0 & 0 & 0 & 0 & 0 & 0 \\
 0 & 0 & 0 & -\frac{1}{6} & 0 & \frac{1}{3} & -\frac{1}{6} & 0 \\
 0 & 0 & 0 & -\frac{1}{6} & 0 & -\frac{1}{6} & \frac{1}{3} & 0 \\
 0 & 0 & 0 & 0 & 0 & 0 & 0 & 0 \\
\end{array}
\right).
\label{eq::hr1}
\end{equation}
When we compute the lower bound \eqref{bo::pptf} we find that $\sigma$ which maximizes it is given by 
\begin{equation}
    \sigma*=\frac{1}{3}\left(\ket{011}\bra{011}+\ket{101}\bra{101}+\ket{110}\bra{110}\right).
    \label{eq::analyt}
\end{equation}
By using $E_g(\rho)\leq 1-F(\rho,\sigma),$ where $\sigma$ is a separable state we show that the geometric entanglement is at most $\frac{1}{3}.$ Lower bound computed by our method is $10^{-9}$ smaller. We thus obtained almost exact value of the geometric entanglement in that case. 
And the geometric entanglement for it is $0.333.$ 
For $h=0$ the ground state is:
\begin{equation}
    \left(
\begin{array}{cccccccc}
 0 & 0 & 0 & 0 & 0 & 0 & 0 & 0 \\
 0 & \frac{1}{6} & -\frac{1}{12} & 0 & -\frac{1}{12} & 0 & 0 & 0 \\
 0 & -\frac{1}{12} & \frac{1}{6} & 0 & -\frac{1}{12} & 0 & 0 & 0 \\
 0 & 0 & 0 & \frac{1}{6} & 0 & -\frac{1}{12} & -\frac{1}{12} & 0 \\
 0 & -\frac{1}{12} & -\frac{1}{12} & 0 & \frac{1}{6} & 0 & 0 & 0 \\
 0 & 0 & 0 & -\frac{1}{12} & 0 & \frac{1}{6} & -\frac{1}{12} & 0 \\
 0 & 0 & 0 & -\frac{1}{12} & 0 & -\frac{1}{12} & \frac{1}{6} & 0 \\
 0 & 0 & 0 & 0 & 0 & 0 & 0 & 0 \\
\end{array}
\right).
\label{hr0}
\end{equation}
The geometric entanglement of this state is $0.25.$
\section{What if the lower bound is higher than the upper bound?}{\label{s::accura}}

\texttt{entcalc} may occasionally output a lower bound that is higher than the upper bound. 
This can happen due to one or more of the following reasons:
\begin{enumerate}
    \item The difference between the lower and upper bound is smaller than the solver precision. 
    In this case, the lower bound lies within an interval centered at the SDP result with a length equal to twice the solver's precision. 
    Increasing the solver precision should eliminate this issue.

    \item The solver stopped before reaching the optimal value. 
    For example, if the maximum number of iterations is too small, the solver may terminate before convergence. 
    SDP solvers typically start from a point that does not satisfy the constraints and gradually converge toward the optimum, improving feasibility and precision at each step. 
    If execution stops prematurely, the result can deviate significantly from the true optimum. 
    In such cases, \texttt{entcalc} reports that the optimum has not been reached.

    \item For some problems, the solver may stop before reaching the true optimum without issuing any warning. 
    This can occur when stopping criteria are met — for instance, when the primal-dual gap, primal feasibility, dual feasibility, or step size drops below predefined thresholds. 
    For numerically unstable states, such as thermal states near quantum phase transitions, it may happen that the stopping criteria are satisfied even though the result is still imprecise. 
    The solver may then not indicate that the required precision was not achieved. 
    Based on our computations, the precision in such cases is typically around \(10^{-7}\). 
    If higher precision is required, one should compare lower bounds computed using both MOSEK and SDPT3 solvers and take the smaller value. 
    It is very rare that both solvers produce imprecise results simultaneously.

    \item The machine precision is typically \(10^{-16}\). 
    When computing the upper bound, a square root operation is applied, effectively reducing the attainable precision to around \(10^{-8}\).
\end{enumerate}

Taking all these factors into account, in the worst-case scenario \texttt{entcalc} should yield results with precision around \(10^{-7}\). 
In most cases, the precision is significantly better. 
However, if the lower bound occasionally exceeds the upper bound by about \(10^{-7}\), this discrepancy is due to numerical precision limitations.
If computations with precision higher than $10^{-7}$ are needed one should use MOSEK API directly.
\end{document}